\documentclass[11pt]{article}
\usepackage{amsmath,amsthm,amsfonts,amscd,eucal,latexsym,amssymb}
\usepackage{mathtools}
\usepackage{hyperref}
\usepackage{lmodern}
\usepackage{epsfig}  
\usepackage{color}
\usepackage[margin=2.5cm]{geometry}
\usepackage[all,cmtip]{xy}
\usepackage{mathrsfs}

\newtheorem{theorem}{Theorem}[section]
\newtheorem{lemma}[theorem]{Lemma}
\newtheorem{proposition}[theorem]{Proposition}
\newtheorem{corollary}[theorem]{Corollary}
\theoremstyle{definition}
\newtheorem{remark}[theorem]{Remark}

\numberwithin{equation}{section}

\def\bC{{\mathbb C}}           
\def\bR{{\mathbb R}} 

\def\bZ{{\mathbb Z}}
\def\bT{{\mathbb T}}
\def\mA{{\mathcal A}}
\def\mB{{\mathcal B}}
\def\mC{{\mathcal C}}

\def\mH{{\mathcal H}}
\def\mI{{\mathcal I}}

\def\mW{{\mathcal W}}

\def\scrD{{\mathscr D}}
\def\scrH{{\mathscr H}}

\def\H{\mathrm{H}}
\def\dH{\hat{\H}}
\def\Ht{\H_\tor}
\def\Hf{\H_\free}
\def\OmegaZ{\Omega_\bZ}
\def\free{\mathrm{free}}
\def\tor{\mathrm{tor}}
\def\sc{\mathrm{sc}}
\def\c{\mathrm{c}}
\def\WF{\mathrm{WF}}
\def\Char{\mathrm{Char}}

\def\Conf{\mathfrak{C}}
\def\TT{\mathfrak{T}}

\def\Dyn{\mathsf{Dyn}}
\def\Top{\mathsf{Top}}
\def\Topf{\Top_\free}
\def\Topt{\Top_\tor}

\def\la{\langle}
\def\ra{\rangle}

\DeclareMathOperator{\cu}{curv}
\DeclareMathOperator{\dcu}{\ast\,\cu}
\DeclareMathOperator{\ch}{char}
\DeclareMathOperator{\dd}{d}
\DeclareMathOperator{\de}{\delta}
\DeclareMathOperator{\tr}{\iota}
\DeclareMathOperator{\fl}{\kappa}
\DeclareMathOperator{\mux}{\mu}
\DeclareMathOperator{\nux}{\nu}
\DeclareMathOperator{\be}{\beta}
\DeclareMathOperator{\jx}{j}
\DeclareMathOperator{\qx}{q}
\DeclareMathOperator{\etax}{\eta}
\DeclareMathOperator{\xix}{\xi}
\DeclareMathOperator{\chix}{\chi}
\DeclareMathOperator{\id}{id}

\begin{document}

\par 
\bigskip 
\par 
\rm 
 
\par 
\bigskip 
\LARGE 
\noindent 
{\bf Hadamard states for quantum Abelian duality} 
\bigskip 
\par 
\rm 
\normalsize 

\large
\noindent {\bf Marco Benini$^{1,a}$},
 {\bf Matteo Capoferri}$^{2,b}$, {\bf Claudio Dappiaggi}$^{3,4,c}$ \\
\par
\small

\noindent $^1$ Institut f\"ur Mathematik, Universit\"at Potsdam,
Karl-Liebknecht-Str.\ 24-25, 14476 Potsdam, Germany.\smallskip

\noindent $^2$  Department of Mathematics, University College London, Gower Street, London WC1E 6BT,
UK. \smallskip

\noindent $^3$ Dipartimento di Fisica, Universit\`a degli Studi di Pavia,
Via Bassi 6, I-27100 Pavia, Italy. \smallskip

\noindent $^4$ Istituto Nazionale di Fisica Nucleare - Sezione di Pavia,
Via Bassi 6, I-27100 Pavia, Italy. \smallskip

\bigskip

\noindent $^a$ mbenini87@gmail.com,
$^b$ matteo.capoferri@gmail.com, $^c$ claudio.dappiaggi@unipv.it
 \normalsize

\par 
 
\rm\normalsize 
\bigskip
\noindent {\small Version of \today}

\rm\normalsize 
 
 
\par 
\bigskip 

\noindent 
\small 
{\bf Abstract}. 
Abelian duality is realized naturally by combining differential cohomology and locally covariant quantum field theory. This leads to a C$^*$-algebra of observables, which encompasses the simultaneous discretization of both magnetic and electric fluxes. We discuss the assignment of physically well-behaved states to such algebra and the properties of the associated GNS triple. We show that the algebra of observables factorizes as a suitable tensor product of three C$^*$-algebras: the first factor encodes dynamical information, while the other two capture topological data corresponding to electric and magnetic fluxes. On the former factor we exhibit a state whose two-point correlation function has the same singular structure of a Hadamard state. Specifying suitable counterparts also on the topological factors we obtain a state for the full theory, providing ultimately a unitary implementation of Abelian duality.

\vskip .3cm


\vskip .3cm


\normalsize
\bigskip


\section{Introduction}

The implementation of the principle of local gauge invariance in the framework of algebraic quantum field theory has been a hot topic in the past years. Besides the obvious connections to models of major physical interest, a notable feature, which has emerged from the earliest investigations, is the violation of the principle of general local covariance \cite{BFV03} in the prime example of an Abelian gauge theory, i.e.\ electromagnetism \cite{DL12,Benini:2013tra,BDHS14,Sanders:2012sf}: it appears to be impossible to reconcile the necessity of gauge invariant dynamical observables with the isotony axiom, i.e.\ injectivity of the morphism between observable algebras induced by a spacetime isometric embedding.

In addition, gauge theories are of paramount interest to various areas of mathematical physics since they are the natural playground for the appearance of dualities, the most famous example being the electric/magnetic duality in source-free Maxwell's theory. In this context, the discrete nature of magnetic fluxes comes from refining the Faraday tensor as the curvature of a circle bundle connection, whereas discrete electric fluxes stem from the so-called Dirac charge quantization. The natural generalization of this mechanism goes under the name of {\em Abelian duality}. 

In this paper we focus on this specific aspect of Abelian gauge theories, further elaborating on the results of \cite{BMath,BPhys}. The starting point of these papers is the reformulation in the framework of locally covariant quantum field theory \cite{BSS14} of the well-established implementation of Abelian gauge theories by means of differential cohomology \cite{Freed:2000ta,Szabo:2012hc}. This viewpoint has been taken a step further in \cite{BPhys} by implementing Abelian duality naturally on globally hyperbolic spacetimes, resulting in a functorial assignment of self-dual gauge fields, which encompasses both the causality and the time-slice axiom. This novel approach has at least two net advantages. On the one hand it yields a simultaneous discretization of electric and magnetic fluxes, which is manifest at the level of the algebra of observables. On the other hand, it enhances the Hamiltonian description by \cite{FMSa,FMSb} in a spacetime covariant fashion: Abelian duality becomes a natural isomorphism between suitable quantum field theory functors. 

The results obtained in \cite{BPhys} still fall one step shorter from a full-fledged description of Abelian duality, namely the assignment of a suitable quantum state that recovers the usual interpretation of quantum theories is missing. Addressing this aspect is the main goal of our paper. While the existence of states is not a matter of debate, not all of them should be considered physically acceptable. Already for the scalar field, in order to guarantee finite quantum fluctuations for all observables as well as the existence of a covariant construction of an algebra of Wick polynomials \cite{KM}, one needs to focus on the restricted class of {\em Hadamard states}. These are characterized by a specific singular structure of the underlying two-point correlation function \cite{Rad96}. Existence, properties and construction schemes for these states have been thoroughly studied in the past twenty years, including a class of examples of Abelian gauge theories \cite{Dappiaggi:2011cj, FP03, BDM, Gerard:2014jba}. In the context of Abelian duality, similar results cannot be obtained blindly due to the non-trivial (and fascinating) entanglement of dynamical and topological degrees of freedom \cite{BPhys}: finding a physically relevant state implementing Abelian duality is the challenging task addressed here. 

In this paper we focus on globally hyperbolic spacetimes of arbitrary dimension, though with compact Cauchy surfaces. For the symplectic Abelian group of observables associated to Abelian duality \cite{BPhys} we provide a (non-canonical) decomposition into three symplectic Abelian subgroups. The first captures dynamical information, while the other two encode the topological degrees of freedom associated to magnetic and electric fluxes. This decomposition carries over upon Weyl quantization: the full C$^*$-algebra of observables is isomorphic to an appropriate tensor product of three C$^*$-algebras, each associated to one of the symplectic Abelian subgroups aforementioned. Therefore, a state for the full C$^*$-algebra is tantamount to one for each tensor factor. Furthermore, having disentangled the topological and the dynamical degrees of freedom, for the latter we can investigate the existence of states with two-point correlation function fulfilling the microlocal spectrum condition \cite{SV01}. With a slight abuse of terminology we will still call them {\em Hadamard states}, although the underlying algebra does not fulfil the standard CCR relations of a scalar field that lie at the core of \cite{Rad96}. To make our analysis more concrete we will construct these Hadamard states explicitly by focusing on ultra-static backgrounds. Additionally, we assign states to the C$^*$-algebras encompassing the topological degrees of freedom. As a by-product, we prove that Abelian duality is unitarily implemented at the level of the GNS triple, thus closing the gap with the existing literature based on a direct Hilbert space description \cite{FMSa,FMSb}.
\vskip 1em

As for the structure of the paper, in Section \ref{Prelimiaries} we introduce our notation and conventions, recollecting in particular the most important results from \cite{BSS14,BPhys,BMath}. Most notably, we discuss the symplectic Abelian group of observables. In Section \ref{Ortho_decomp}, we introduce three auxiliary symplectic Abelian groups, corresponding to the dynamical sector, to the torsion-free topological sector and to the torsion topological sector of our model. These are then used in Section \ref{secDecomposition} to present the symplectic Abelian group of observables as a direct sum of the mentioned sectors. In Section \ref{Sec:quantstates}, firstly we recall that we can associate a C$^*$-algebra of Weyl type to each symplectic Abelian group. Secondly we prove that the splitting of the preceding section entails a factorization of the full C$^*$-algebra of observables into a suitable tensor product of C$^*$-algebras, each factor being associated to one of the sectors. Focusing on the case of ultra-static, globally hyperbolic spacetimes with compact Cauchy surface, in Section \ref{Sec:states_dyn} we construct the ground state for the C$^*$-algebra associated to the dynamical sector, proving that its two-point correlation function fulfils the microlocal spectrum condition. Then we exhibit states also for the remaining two sectors, see Sections \ref{sec:TopfState} and \ref{sec:ToptState}, thus showing that Abelian duality is unitarily implemented on the GNS triple. Section \ref{sec:LorCyl} concludes the paper discussing the example of the Lorentz cylinder, providing a Fourier expansion of the two-point correlation function for the dynamical sector of the theory: remarkably, Abelian duality naturally circumvents the infrared obstructions to the existence of ground states in 1+1 dimensions.


\section{Preliminaries}\label{Prelimiaries}
In the present section we quickly recapitulate the background material for the rest of the paper. 
We also take the chance to introduce our notation and conventions. 
Let us remark once and for all that, unless otherwise stated, for any $m \in \bZ_{\geq0}$, 
the term \textit{($m$-dimensional) manifold} refers to a connected, second-countable, 
Hausdorff topological space that is locally homeomorphic to $\bR^n$ 
and that comes equipped with a smooth structure (an atlas with smooth transition maps). 
We will only consider manifolds of \textit{finite type}, namely those admitting a finite good cover. 
For Abelian groups we will always use the additive notation, 
for example $\bT$ is defined as the cokernel of the homomorphism $\bZ \to \bR$ of Abelian groups, 
while the multiplicative notation is reserved to ring multiplication, 
e.g.\ $\wedge$ for the graded algebra of differential forms, 
$\smile$ for the graded ring structure on cohomology, 
$\cdot$ for the graded differential cohomology ring 
and, more generally, juxtaposition for associative algebras.

\subsection{Differential cohomology}
The mathematical tool to describe configurations for the field theory under analysis 
is differential cohomology. For a full presentation of the subject 
we refer the reader to the existing literature, e.g.\ \cite{BB}. 
Here we will just recall few basic facts, mainly in order to introduce our notation. 
Differential cohomology for a manifold $M$ is a graded ring $\dH^\ast(M;\bZ)$ 
that arises as a suitable refinement of the cohomology ring $\H^\ast(M;\bZ)$ with $\bZ$-coefficients 
by the graded algebra $\Omega^\ast(M)$ of differential forms. 
More precisely, for $k \in \bZ_{\geq0}$, $k$-differential cohomology 
is a contravariant functor $\dH^k(\,\cdot\,;\bZ)$ 
defined on a suitable category of ``spaces'' (including manifolds) and taking values in Abelian groups. 
Such functor is equipped with four natural transformations $\tr$, $\fl$, $\cu$, $\ch$ 
relating it to differential forms and cohomology classes according to the commutative diagram depicted below, 
whose rows and columns are short exact sequences:
\begin{equation}\label{diaDiffCoho}
\xymatrix{
		&	0 \ar[d]	&	0 \ar[d]	&	0\ar[d]																\\
0 \ar[r]	&	\frac{\H^{k-1}(M;\bR)}{\Hf^{k-1}(M;\bZ)} \ar[r]^-{\nux} \ar[d]_-{\mux}
					&	\frac{\Omega^{k-1}(M)}{\OmegaZ^{k-1}(M)} \ar[r]^-\dd \ar[d]_-\tr
								&	\dd\Omega^{k-1}(M) \ar[r] \ar[d]^-{\subseteq}				&	0	\\
0 \ar[r]	&	\H^{k-1}(M;\bT) \ar[r]^-\fl \ar[d]_-\be
					&	\dH^k(M;\bZ) \ar[r]^-\cu \ar[d]_-\ch
								&	\OmegaZ^k(M) \ar[r] \ar[d]^-{[\,\cdot\,]}					&	0	\\
0 \ar[r]	&	\Ht^k(M;\bZ) \ar[r]_-\jx \ar[d]
					&	\H^k(M;\bZ) \ar[r]_-\qx \ar[d]
								&	\Hf^k(M;\bZ) \ar[r] \ar[d] 									&	0	\\
		&	0		&	0		&	0
}
\end{equation}
There are several equivalent models for differential cohomology, 
e.g. Cheeger-Simons differential characters \cite{CS}, de Rham-Federer characters \cite{HLZ}, 
Hopkins-Singer cocycles \cite{HS}. Furthermore, differential cohomology classifies isomorphism classes 
of higher circle bundles equipped with a connection. This interpretation of differential cohomology 
motivates the names usually attributed to the natural transformations mentioned above, 
namely characteristic class for $\ch$, curvature for $\cu$, 
inclusion of flat connections for $\fl$ and inclusion of trivial bundles for $\tr$. 
In analogy with cohomology and differential forms, differential cohomology 
is a graded ring with multiplication $\,\cdot\,$. 
Both $\ch$ and $\cu$ are ring homomorphisms with respect to $\,\cdot\,$, 
while $\tr$ and $\fl$ are compatible with the ring structure as indicated below: 
\begin{align}\label{eqRingComp}
\tr [A] \cdot h = \tr [A \wedge \cu h],		&&	\fl u \cdot h = \fl (u \smile \ch h), 
\end{align}
for all $h \in \dH^k(M;\bZ)$, $[A] \in \Omega^\ell(M) / \OmegaZ^\ell(M)$ and $u \in \H^\ell(M;\bT)$. 
Further information about differential cohomology and its graded ring structure can be found in \cite{BB,SS}. 
Let us also mention that relative versions of differential cohomology exist 
(see \cite{BB} for a comparison among different approaches)
and these can be used to realize differential cohomology with restricted support, 
for example models with compact support have been considered in \cite{HLZ, BMath}.

\subsection{Configurations}
By means of differential cohomology in degree $k$ 
on an $m$-dimensional globally hyperbolic spacetime\footnote{Here 
and in the following, the term \textit{spacetime} refers to an oriented and time-oriented Lorentzian manifold.} 
$M$ we can introduce the model under investigation, specifying its equation of motion 
and solving the associated Cauchy problem. This has been developed in \cite{BPhys}; 
however, for ease of reference we recapitulate the main results below. 
\vskip 1em

Denoting with $\ast$ the Hodge dual induced by the metric and by the orientation of $M$, we specify the Abelian group 
\begin{equation}\label{eqConf}
\Conf^k(M;\bZ) \doteq \big\{ (h,\tilde h) \in \dH^k(M;\bZ) \times \dH^{m-k}(M;\bZ):\; \cu h = \dcu \tilde h \big\},
\end{equation}
encompassing the configurations $(h,\tilde h)$ for the field theory of interest. 
In degree $k=2$ and dimension $m=4$, this provides a \textit{semiclassical} refinement 
of Maxwell theory in the vacuum and without external sources, with the interesting feature 
of encoding the discretization of electric and magnetic fluxes 
that originate from topological features of the underlying spacetime. 
This refinement has the pleasant feature 
of preserving the duality between electric and magnetic degrees of freedom 
typical of the source-free Maxwell theory, see \cite{FMSa, FMSb, BPhys}. 
In fact, recalling that $\ast \ast = (-1)^{p(m-p)+1}$ on $p$-forms over $M$, one obtains a natural isomorphism 
\begin{align}\label{eqn:duality}
\zeta: \Conf^k(M;\bZ) \longrightarrow \Conf^{m-k}(M;\bZ), 
&& (h, \tilde h) \longmapsto (\tilde h, (-1)^{k(m-k)+1} h),
\end{align}
that interchanges precisely the (discrete) magnetic and electric fluxes, 
carried respectively by $h$ and $\tilde h$ and detected through $\ch$. 
In particular, for $m=2k$ this duality is a natural automorphism 
that can be used to identify self-dual configurations \cite{BPhys}. 
\vskip 1em

The configuration space $\Conf^k(M;\bZ)$ fits into a commutative diagram 
whose rows and columns form short exact sequences, see \cite[Sect.\ 2.1]{BPhys}. 
For any $p,q \in \bZ_{\geq0}$ and any graded Abelian group $A$ set 
\begin{align}
A^{p,q} = A^p \times A^q
\end{align}
and introduce topologically trivial configurations by 
\begin{equation}
\TT^k(M;\bZ) \doteq \big\{ ([A],[\tilde A]) \in \Omega^{k-1,m-k-1}(M) / \OmegaZ^{k-1,m-k-1}(M):\; 
						\dd A = \ast \dd \tilde A \big\}. 
\end{equation}
Then the above-mentioned diagram reads as follows: 
\begin{equation}\label{diaConf}
\xymatrix@C=2.5em{
		&	0 \ar[d]		&	0 \ar[d]		&	0\ar[d]														\\
0 \ar[r]	&	\frac{\H^{k-1,m-k-1}(M;\bR)}{\Hf^{k-1,m-k-1}(M;\bZ)}
				\ar[r]^-{\nux \times \nux} \ar[d]_-{\mux \times \mux}
						&	\TT^k(M;\bZ) \ar[r]^-{\dd_1} \ar[d]_-{\tr \times \tr}
										&	\dd\Omega^{k-1} \cap \ast\,\dd\Omega^{m-k-1}(M) 
												\ar[r] \ar[d]^-\subseteq							&	0	\\
0 \ar[r]	&	\H^{k-1,m-k-1}(M;\bT) \ar[r]^-{\fl \times \fl} \ar[d]_-{\be \times \be}
						&	\Conf^k(M;\bZ) \ar[r]^-{\cu_1} \ar[d]_-{\ch \times \ch}
										&	\OmegaZ^k \cap \ast\,\OmegaZ^{m-k}(M)
												\ar[r] \ar[d]^-{([\,\cdot\,],[\ast^{-1}\,\cdot\,])}	&	0	\\
0 \ar[r]	&	\Ht^{k,m-k}(M;\bZ) \ar[r]_-{\jx \times \jx} \ar[d]
						&	\H^{k,m-k}(M;\bZ) \ar[r]_-{\qx \times \qx} \ar[d]
										&	\Hf^{k,m-k}(M;\bZ) \ar[r] \ar[d] 						&	0	\\
		&	0			&	0			&	0
}
\end{equation}
where the subscript $_1$ denotes the precomposition with the projection on the first factor. 
Using the diagram above, homotopy invariance of cohomology groups 
and solving the Cauchy problem for Maxwell's equation \cite{DL12}, 
one obtains the solution to the Cauchy problem formalized below \cite[Sect.\ 2.2]{BPhys}: 
for a spacelike Cauchy surface $\Sigma$ of $M$ 
and for arbitrary initial data $(h_\Sigma,\tilde h_\Sigma) \in \dH^{k,m-k}(\Sigma;\bZ)$, 
there exists a unique configuration $(h, \tilde h) \in \Conf^k(M;\bZ)$ 
that restricts to the given initial data along the embedding $i_\Sigma: \Sigma \to M$, i.e.\ 
\begin{align}
\cu h = \dcu \tilde h,	&&	i_\Sigma^\ast h = h_\Sigma,	&&	i_\Sigma^\ast \tilde h = \tilde h_\Sigma, 
\end{align}
where $i_\Sigma^\ast \doteq \dH^p(i_\Sigma;\bZ): \dH^p(M;\bZ) \to \dH^p(\Sigma;\bZ)$ 
denotes the differential cohomology pullback. This means that the map 
\begin{align}\label{eqCauchyIso}
i_\Sigma^\ast: \Conf^k(M;\bZ) \longrightarrow \dH^{k,m-k}(\Sigma;\bZ), 
	&&	(h, \tilde h) \longmapsto (i_\Sigma^\ast h, i_\Sigma^\ast \tilde h) 
\end{align}
is an isomorphism of Abelian groups. 
A similar result holds true for all objects in diagram \eqref{diaConf}, which becomes isomorphic to 
\begin{equation}\label{diaInitialData}
\xymatrix@C=3.5em{
		&	0 \ar[d]	&	0 \ar[d]	&	0\ar[d]																	\\
0 \ar[r]	&	\frac{\H^{k-1,m-k-1}(\Sigma;\bR)}{\Hf^{k-1,m-k-1}(\Sigma;\bZ)}
				\ar[r]^-{\nux \times \nux} \ar[d]_-{\mux \times \mux}
					&	\frac{\Omega^{k-1,m-k-1}(\Sigma)}{\OmegaZ^{k-1,m-k-1}(\Sigma)} 
							\ar[r]^-{\dd \times \dd} \ar[d]_-{\tr \times \tr}
								&	\dd\Omega^{k-1,m-k-1}(\Sigma) \ar[r] \ar[d]^-{\subseteq}		&	0	\\
0 \ar[r]	&	\H^{k-1,m-k-1}(\Sigma;\bT) \ar[r]^-{\fl \times \fl} \ar[d]_-{\be \times \be}
					&	\dH^{k,m-k}(\Sigma;\bZ) \ar[r]^-{\cu \times \cu} \ar[d]_-{\ch \times \ch}
								&	\OmegaZ^{k,m-k}(\Sigma) \ar[r] \ar[d]^-{([\,\cdot\,],[\,\cdot\,])}	&	0	\\
0 \ar[r]	&	\Ht^{k,m-k}(\Sigma;\bZ) \ar[r]_-{\jx \times \jx} \ar[d]
					&	\H^{k,m-k}(\Sigma;\bZ) \ar[r]_-{\qx \times \qx} \ar[d]
								&	\Hf^{k,m-k}(\Sigma;\bZ) \ar[r] \ar[d] 							&	0	\\
		&	0		&	0		&	0
}
\end{equation}
upon pullback along $i_\Sigma: \Sigma \to M$. Therefore, the assignment of a spacelike Cauchy surface 
provides an equivalent way to describe the configuration space $\Conf^k(M;\bZ)$. 
We will often adopt this point of view in the following, especially in Section \ref{Ortho_decomp}. 
\vskip 1em

\begin{remark}
Similar conclusions hold true for a modified version $\Conf_\sc^k(M;\bZ)$ 
of the configuration space $\Conf^k(M;\bZ)$ with support restricted to spacelike compact regions, 
see \cite[Sect.\ 3]{BPhys}. In particular, the embedding $i_\Sigma: \Sigma \to M$ 
of a spacelike Cauchy surface into the globally hyperbolic spacetime $M$ 
induces an isomorphism similar to \eqref{eqCauchyIso}: 
\begin{equation}\label{eqSCCauchyIso}
i_\Sigma^\ast: \Conf_\sc^k(M;\bZ) \longrightarrow \dH_\c^{k,m-k}(\Sigma;\bZ), 
\end{equation}
where $\dH_\c^p(\Sigma;\bZ)$ is the Abelian group of $p$-differential characters 
with compact support on $\Sigma$ (cf.\ \cite{BMath,HLZ}). 
The Abelian group $\Conf_\sc^k(M;\bZ)$ of spacelike compact configurations 
plays an important role because it can be used to introduce a class of well-behaved functionals 
on $\Conf^{k,m-k}(M;\bZ)$ that can be regarded as \textit{observables} for this theory, \cite[Sect.\ 4.2]{BPhys}. 
Notice that diagrams \eqref{diaConf} and \eqref{diaInitialData} have counterparts 
with spacelike compact and, respectively, compact support (see \cite{BMath,BPhys}) 
that provide crucial information about the observables of the model under consideration. 
\end{remark}

\subsection{\label{secObs}Observables}
In this section we recall the notion of observable considered in \cite[Sect.\ 4]{BPhys}. 
Throughout the paper we will be dealing 
with globally hyperbolic spacetimes $M$ admitting a \textit{compact} Cauchy surface; 
for this reason, from now on and unless otherwise stated we will focus on this specific case, 
thus avoiding some of the technical complications that arise in the general situation, cf.\ \cite{BPhys}. 
Compared with the general case, the peculiar feature of having a compact Cauchy surface is that 
there are no degeneracies in the presymplectic structure, i.e.\ it is symplectic, cf.\ \cite[Prop.\ 4.5]{BPhys}. 
\vskip 1em

Exploiting the fact that $M$ admits compact Cauchy surfaces, 
we introduce a non-degenerate pairing $\sigma: \Conf^k(M;\bZ) \times \Conf^k(M;\bZ) \to \bT$ 
and we use it to interpret $\Conf^k(M;\bZ)$ as the Abelian group 
labelling a well-behaved class of observables on $\Conf^k(M;\bZ)$ 
of the form $\sigma(\,\cdot\, ,(h,\tilde h))$, $(h,\tilde h) \in \Conf^k(M;\bZ)$. 
At the same time, this pairing will also provide a symplectic structure on $\Conf^k(M;\bZ)$. 
The procedure to define $\sigma$ involves the isomorphism \eqref{eqCauchyIso} 
between configurations and initial data on a spacelike Cauchy surface $\Sigma$ of $M$ 
and the definition of a suitable pairing $\sigma^\Sigma$ on initial data $\dH^{k,m-k}(\Sigma;\bZ)$. 
To introduce $\sigma^\Sigma: \dH^{k,m-k}(\Sigma;\bZ) \times \dH^{k,m-k}(\Sigma;\bZ) \to \bT$ 
we proceed as follows:
\begin{enumerate}
\item Recalling the ring structure for differential cohomology on $\Sigma$, 
we introduce a bi-homomorphism of Abelian groups 
\begin{align}
\dH^{k,m-k}(\Sigma;\bZ) \times \dH^{k,m-k}(\Sigma;\bZ) & \longrightarrow \dH^m(\Sigma;\bZ), \nonumber \\
\big( (h_\Sigma,\tilde h_\Sigma), (h^\prime_\Sigma,\tilde h^\prime_\Sigma) \big) 
& \longmapsto \tilde h_\Sigma \cdot h_\Sigma^\prime - \tilde h_\Sigma^\prime \cdot h_\Sigma;
\end{align}
\item Due to $\dim \Sigma = m-1$, it follows that $\dH^m(\Sigma;\bZ)$ 
is isomorphic to $\H^{m-1}(\Sigma;\bR)/\Hf^{m-1}(\Sigma;\bZ)$ in a natural way; 
\item $\Sigma$ being compact and oriented (its orientation is induced by orientation and time-orientation of $M$), 
we can consider its fundamental class $[\Sigma] \in H_{m-1}(\Sigma)$; 
\item $\Sigma$ is also connected, hence the canonical evaluation of cohomology classes on $[\Sigma]$ yields an isomorphism $\H^{m-1}(\Sigma;\bR)/\Hf^{m-1}(\Sigma;\bZ) \to \bT$. 
\end{enumerate}
These considerations lead to the definition of $\sigma^\Sigma$ given below: 
\begin{align}\label{eqn:symplformSigma}
\sigma^\Sigma: \dH^{k,m-k}(\Sigma;\bZ) \times \dH^{k,m-k}(\Sigma;\bZ) \longrightarrow \bT, 
	&&	\big( (h_\Sigma,\tilde h_\Sigma), (h^\prime_\Sigma,\tilde h^\prime_\Sigma) \big) \longmapsto 
			(\tilde h_\Sigma \cdot h_\Sigma^\prime - \tilde h_\Sigma^\prime \cdot h_\Sigma)[\Sigma]. 
\end{align} 
Notice that by \cite[Prop.\ 5.6]{BMath} $\sigma^\Sigma$ is a weakly non-degenerate pairing; 
moreover, it is antisymmetric. Pulling $\sigma^\Sigma$ back along the isomorphism 
$i_\Sigma^\ast: \Conf^k(M;\bZ) \to \dH^{k,m-k}(\Sigma;\bZ)$ in \eqref{eqCauchyIso} we obtain 
\begin{equation}\label{eqn:symplform}
\sigma \doteq \sigma^\Sigma \circ (i_\Sigma^\ast \times i_\Sigma^\ast): 
	\Conf^k(M;\bZ) \times \Conf^k(M;\bZ) \longrightarrow \bT. 
\end{equation}
An argument based on Stokes' theorem shows that $\sigma$ does not depend 
on the chosen (compact) spacelike Cauchy surface $\Sigma$, cf.\ \cite[Lem.\ 8.4]{BPhys}. 
Clearly, $\sigma$ inherits the properties of $\sigma^\Sigma$, 
in particular it is a weakly non-degenerate antisymmetric pairing. 

For $(h,\tilde h) \in \Conf^k(M;\bZ)$, $\sigma(\,\cdot\, ,(h,\tilde h))$ defines a functional 
on $\Conf^k(M;\bZ)$. \cite{BPhys, BMath} show that these functionals are distinguished, 
in the sense that they provide smooth characters on the configuration space. 
We interpret such functionals as the observables for the model under consideration. 
Since $\sigma$ is weakly non-degenerate, functionals of this type 
form an Abelian group isomorphic to $\Conf^k(M;\bZ)$. 
Furthermore, $\sigma$ is antisymmetric, therefore we regard 
\begin{equation}
(\Conf^k(M;\bZ),\sigma)
\end{equation}
as the symplectic space of observables for our theory. 
Recalling \eqref{eqn:symplformSigma} and \eqref{eqn:symplform} and on account of graded commutativity, 
it follows that the duality isomorphism $\zeta: \Conf^k(M;\bZ) \to \Conf^{m-k}(M;\bZ)$ 
defined in \eqref{eqn:duality} preserves $\sigma$, 
hence Abelian duality is implemented symplectically at the level of observables.

\section{Symplectically orthogonal decomposition}\label{Ortho_decomp}
The goal of this section is to establish a decomposition of the Abelian group of observables $\Conf^k(M;\bZ)$ 
that is compatible with its symplectic structure $\sigma$. 
As we will see later, at the level of quantum algebras this decomposition corresponds 
to a convenient factorization that will enable us to introduce a state on the full algebra 
by looking at each factor independently. 
Once again, we consider globally hyperbolic spacetimes $M$ admitting a compact Cauchy surface. 
The peculiarity of the compact Cauchy surface case is that all pairings we are going to consider 
are weakly non-degenerate, while they happen to have degeneracies in general. 
\vskip 1em

Let $M$ be an $m$-dimensional globally hyperbolic spacetime having a compact spacelike Cauchy surface $\Sigma$. 
As above, we will denote the embedding of $\Sigma$ into $M$ by $i_\Sigma$ 
and we will often implicitly consider the isomorphism 
$i_\Sigma^\ast: \Conf^k(M;\bZ) \to \dH^{k,m-k}(\Sigma;\bZ)$ of \eqref{eqCauchyIso}. 
In particular, recall that $i_\Sigma^\ast$ extends to an isomorphism 
between the diagrams \eqref{diaConf} and \eqref{diaInitialData} 
and that it relates the symplectic structure $\sigma$ on $\Conf^k(M;\bZ)$ 
to the symplectic structure $\sigma^\Sigma$ on $\dH^{k,m-k}(\Sigma;\bZ)$. 
Our task is twofold: first, we aim at introducing suitable symplectic structures 
on certain auxiliary Abelian groups arising from the diagrams \eqref{diaConf} and \eqref{diaInitialData}; 
second, we want to provide a symplectically orthogonal decomposition 
of the symplectic Abelian group $(\Conf^k(M;\bZ),\sigma)$ 
in terms of the above-mentioned auxiliary symplectic Abelian groups.

\subsection{\label{secDyn}Dynamical sector}
With \textit{dynamical sector} we refer to the top-right corner of diagram \eqref{diaConf}, i.e.\ the vector space 
\begin{equation}
\Dyn^k(M) \doteq \dd \Omega^{k-1}(M) \cap \ast \dd \Omega^{m-k-1}(M). 
\end{equation}
The reader should keep in mind that, as a special case of \cite[Th.\ 2.5]{BPhys}, 
one obtains an isomorphism induced by the restriction to the spacelike Cauchy surface $\Sigma$ of $M$: 
\begin{align}\label{eqDynIso}
i_\Sigma^\ast: \Dyn^k(M) \overset{\simeq}{\longrightarrow} \dd \Omega^{k-1,m-k-1} (\Sigma), 
	&&	\dd A = \ast \dd \tilde A \longmapsto (i_\Sigma^\ast \dd A, i_\Sigma^\ast \dd \tilde A). 
\end{align}
In the following we will often omit to spell out both ways to express an element of $\Dyn(M)$, 
that is to say that we will only present it as $\dd A \in \Dyn^k(M)$, 
although by definition there exists $\tilde A \in \Omega^{m-k-1}(M)$ such that $\dd A = \ast \dd \tilde A$. 
\vskip 1em

We can endow $\Dyn(M)$ with a pairing, namely 
\begin{align}\label{eqDynSympl}
\sigma_\Dyn: \Dyn^k(M) \times \Dyn^k(M) \longrightarrow \bR, 
	&&	(\dd A, \dd A^\prime) \longmapsto \int_\Sigma i_\Sigma^\ast 
			\big( \tilde A \wedge \dd A^\prime - \tilde A^\prime \wedge \dd A \big). 
\end{align}
Stokes theorem implies that $\sigma_\Dyn$ is well-defined 
and that its definition does not depend on the choice of $\Sigma$. 
Clearly, there is an equivalent pairing $\sigma_\Dyn^\Sigma$ 
on the isomorphic vector space $\dd_\Sigma \Omega^{k-1,m-k-1}(\Sigma)$ 
and explicitly defined using the same formula 
(note the subscript $_\Sigma$ to distinguish the differential on $\Sigma$ from the one on $M$): 
\begin{equation}
\sigma_\Dyn^\Sigma \big( (\dd_\Sigma A_\Sigma, \dd_\Sigma A^\prime_\Sigma), 
	(\dd_\Sigma \tilde A_\Sigma, \dd_\Sigma \tilde A^\prime_\Sigma) \big) 
		= \int_\Sigma \tilde A_\Sigma \wedge \dd_\Sigma A^\prime_\Sigma 
			- \tilde A^\prime_\Sigma \wedge \dd_\Sigma A_\Sigma, 
\end{equation}
for $(\dd_\Sigma A_\Sigma, \dd_\Sigma A^\prime_\Sigma), 
(\dd_\Sigma \tilde A_\Sigma, \dd_\Sigma \tilde A^\prime_\Sigma) \in \dd_\Sigma \Omega^{k-1,m-k-1}(\Sigma)$. 
From its definition and using Stokes theorem, 
one can conclude that $\sigma_\Dyn^\Sigma$ is both weakly non-degenerate and anti-symmetric. 
Therefore, $(\dd_\Sigma \Omega^{k-1,m-k-1}(\Sigma), \sigma_\Dyn^\Sigma)$ 
and $(\Dyn(M),\sigma_\Dyn)$ are isomorphic symplectic vector spaces. 
\vskip 1em

The dynamical sector carries a natural duality isomorphism, the counterpart of \eqref{eqn:duality}, 
which evidently preserves the symplectic structure $\sigma_\Dyn$: 
\begin{align}\label{eqn:dualityDyn}
\zeta_\Dyn: \Dyn^k(M) \longrightarrow \Dyn^{m-k}(M), 
&& \dd A = \ast \dd \tilde A \longmapsto \dd \tilde A = \ast \dd (-1)^{k(m-k)+1} A. 
\end{align}

\begin{remark}
We just established a symplectic structure $\sigma_\Dyn$ for the vector space $\Dyn^k(M)$, 
the top-right corner of diagram \eqref{diaConf}; moreover, we introduced an isomorphic symplectic vector space 
$(\dd \Omega^{k-1,m-k-1}(\Sigma), \sigma_\Dyn^\Sigma)$ 
corresponding to initial data on a (compact) spacelike Cauchy surface $\Sigma$. 
Recall that $(\Conf^k(M;\bZ),\sigma)$ is just a symplectic Abelian group, 
rather than a symplectic vector space; 
therefore, in order to relate $(\Dyn^k(M),\sigma_\Dyn)$ to $(\Conf^k(M;\bZ),\sigma)$, 
we will have to forget the multiplication by scalars in $\Dyn^k(M)$ 
and turn $\sigma_\Dyn$ into a $\bT$-valued bi-homomorphism, 
i.e.\ post-compose it with the quotient $\bR \to \bT = \bR/\bZ$. 
We will do so in Section \ref{secDecomposition} in order to identify $(\Dyn^k(M),\sigma_\Dyn)$ 
as a direct summand of $(\Conf^k(M;\bZ),\sigma))$. 
\end{remark}

\subsection{\label{secTopf}Torsion-free topological sector}
The second contribution we consider for our decomposition of $\Conf^k(M;\bZ)$ 
contains topological information only. We quotient out the part related to torsion subgroups 
as those have to be treated independently (and have their own interpretation, cf.\ \cite{FMSb}). 
This leads to the \textit{torsion-free topological sector}, namely the Abelian group arising from the direct sum 
of the top-left and bottom-right corners of diagram \eqref{diaConf}: 
\begin{equation}\label{Topf}
\Topf^k(M) \doteq \frac{\H^{k-1,m-k-1}(M;\bR)}{\Hf^{k-1,m-k-1}(M;\bZ)} \oplus \Hf^{k,m-k}(M;\bZ). 
\end{equation}
Also $\Topf^k(M)$ has an equivalent description in terms of data 
specified on a (compact) spacelike Cauchy surface $\Sigma$: 
\begin{equation}\label{eqTopfIso}
i_\Sigma^\ast: \Topf^k(M) \overset{\simeq}{\longrightarrow}
	\frac{\H^{k-1,m-k-1}(\Sigma;\bR)}{\Hf^{k-1,m-k-1}(\Sigma;\bZ)} \oplus \Hf^{k,m-k}(\Sigma;\bZ). 
\end{equation}
This isomorphism, obtained by pulling cohomology classes back along $i_\Sigma: \Sigma \to M$, 
is just an instance of homotopy invariance of cohomology groups.
Further information can be found in \cite[Lem.\ 8.2]{BPhys}. 
\vskip 1em

We equip $\Topf^k(M)$ with a pairing $\sigma_\free$ 
induced by the cup product $\smile$ between cohomology groups. 
The procedure to define $\sigma_\free$ is similar to the one adopted in Section \ref{secObs}, 
namely we introduce a pairing $\sigma_\free^\Sigma$ on the right-hand-side of \eqref{eqTopfIso} 
and then we induce $\sigma_\free$ on $\Topf^k(M)$ via the isomorphism \eqref{eqTopfIso}. 
As a first step, following \cite[(5.22)]{BMath}, 
we observe that the cohomological cup product provides a bi-homomorphism of Abelian groups:
\begin{align}\label{eqTZCup}
\H^p(\Sigma;\bT) \times \H^{m-p-1}(\Sigma;\bZ) \longrightarrow \H^{m-1}(\Sigma;\bT), 
	&&	(f_\Sigma,\zeta_\Sigma) \longmapsto f_\Sigma \smile \zeta_\Sigma. 
\end{align}
Note that for dimensional reasons $\H^{m-1}(\Sigma;\bT)$ is naturally isomorphic 
to the quotient of $\H^{m-1}(\Sigma;\bR)$ by its subgroup $\Hf^{m-1}(\Sigma;\bZ)$, cf.\ \eqref{diaDiffCoho}. 
Therefore, repeating arguments 3.\ and 4.\ of Section \ref{secObs} 
and with the help of diagram \eqref{diaDiffCoho}, we can introduce the non-degenerate pairing 
\begin{align}\label{eqFreePairing}
\la \cdot,\cdot \ra_\free: \frac{\H^p(\Sigma;\bR)}{\Hf^p(\Sigma;\bZ)} \times \Hf^{m-p-1}(\Sigma;\bZ) 
		\longrightarrow \bT, 
	&&	(u_\Sigma,z_\Sigma) \longmapsto (\mux u_\Sigma \smile \zeta_\Sigma)[\Sigma], 
\end{align}
where$\zeta_\Sigma$ is any element of $\H^{m-p-1}(\Sigma;\bZ)$ such that $\qx \zeta_\Sigma = z_\Sigma$.  
This definition is well-posed on account of the properties of the cup product. 
A suitable combination of the pairings $\la \cdot,\cdot \ra_\free$ for different degrees provides 
\begin{subequations}\label{eqTopfSymplSigma}
\begin{equation}
\sigma_\free^\Sigma: \left(\frac{\H^{k-1,m-k-1}(\Sigma;\bR)}{\Hf^{k-1,m-k-1}(\Sigma;\bZ)} \oplus \Hf^{k,m-k}(\Sigma;\bZ) \right) \times \left( \frac{\H^{k-1,m-k-1}(\Sigma;\bR)}{\Hf^{k-1,m-k-1}(\Sigma;\bZ)} \oplus \Hf^{k,m-k}(\Sigma;\bZ) \right) \longrightarrow \bT, 
\end{equation}
where 
\begin{align}
\sigma_\free^\Sigma & \big( (u,\tilde u,z,\tilde z), (u^\prime,\tilde u^\prime,z^\prime,\tilde z^\prime) \big) 	\\
	&	\doteq \la \tilde u_\Sigma, z^\prime_\Sigma \ra_\free 
			- (-1)^{k(m-k)} \la u_\Sigma, \tilde z^\prime_\Sigma \ra_\free 
			- \la \tilde u^\prime_\Sigma, z_\Sigma \ra_\free 
			+ (-1)^{k(m-k)} \la u^\prime_\Sigma, \tilde z_\Sigma \ra_\free. 
	\nonumber 
\end{align}
\end{subequations}
$\sigma_\free^\Sigma$ is clearly antisymmetric and inherits non-degeneracy 
from $\la \cdot,\cdot \ra_\free$. These properties are directly transferred 
to the pairing $\sigma_\free$, defined on $\Topf^k(M)$ 
as the pullback of $\sigma_\free^\Sigma$ along \eqref{eqTopfIso}: 
\begin{equation}\label{eqTopfSympl}
\sigma_\free \doteq \sigma_\free^\Sigma \circ (i_\Sigma^\ast \times i_\Sigma^\ast): 
	\Topf^k(M) \times \Topf^k(M) \longrightarrow \bT. 
\end{equation}
Notice that $\sigma_\free$ is actually independent of the choice of $\Sigma$. 
In fact, for any choice of Cauchy surface $\Sigma$, 
$i_{\Sigma\,\ast}[\Sigma]$ is the unique generator of $\H_{m-1}(M) \simeq \bZ$. 
Summing up, the right-hand-side of \eqref{eqTopfIso} equipped with $\sigma_\free^\Sigma$ 
and $(\Topf^k(M),\sigma_\free)$ are isomorphic symplectic Abelian groups. 
In Section \ref{secDecomposition} we will identify $(\Topf^k(M),\sigma_\free)$ 
as a direct summand of the symplectic Abelian group $(\Conf^k(M;\bZ),\sigma)$. 
\vskip 1em

As the dynamical sector, also $(\Topf^k(M),\sigma_\free)$ carries 
a counterpart of the duality isomorphism \eqref{eqn:duality}: 
\begin{align}\label{eqn:dualityFree}
\zeta_\free: \Topf^k(M) \longrightarrow \Topf^{m-k}(M), 
&& (u,\tilde u,z,\tilde z) \longmapsto (\tilde u, (-1)^{k(m-k)+1} u, \tilde z, (-1)^{k(m-k)+1} z). 
\end{align}
Note that the one above is a natural isomorphism preserving $\sigma_\free$, 
hence duality is symplectically implemented also on the torsion-free topological sector.

\subsection{Torsion topological sector}
This is the last contribution we have to consider in order to decompose $\Conf^k(M;\bZ)$. 
Again it contains information of purely topological nature, but it is quite special 
in that it relates to the torsion part of certain cohomology groups. 
An interpretation of these quantities in terms of non-commutativity 
between electric and magnetic fluxes can be found in \cite{FMSa, FMSb}. 
As in the previous sections, we will provide two equivalent ways 
to describe the object of interest, related by an isomorphism induced by the embedding $i_\Sigma: \Sigma \to M$ 
of a (compact) spacelike Cauchy surface $\Sigma$ into the globally hyperbolic spacetime $M$. 
Subsequently, we will introduce a suitable symplectic structure. 
We will refer to the Abelian group 
\begin{equation}
\Topt^k(M) \doteq \Ht^{k,m-k}(M;\bZ) 
\end{equation}
as the \textit{torsion topological sector}. 
Since the embedding $i_\Sigma: \Sigma \to M$ is a retraction, 
homotopy invariance of cohomology implies that the restriction along $i_\Sigma$ 
induces an equivalent description of $\Topt^k(M)$ in terms of cohomology groups of the Cauchy surface $\Sigma$: 
\begin{equation}\label{eqToptIso}
i_\Sigma^\ast: \Topt^k(M) \overset{\simeq}{\longrightarrow} \Ht^{k,m-k}(\Sigma;\bZ). 
\end{equation}
Notice that this is an isomorphism of Abelian groups in contrast to \eqref{eqDynIso}, 
which is an isomorphism of vector spaces. 
\vskip 1em

Adopting the Cauchy surface point of view on $\Topt^k(M)$, it is easy to construct a symplectic structure, 
that essentially arises from the torsion linking form. 
We introduce a non-degenerate $\bT$-valued pairing between $\Ht^p(\Sigma;\bZ)$ 
and $\Ht^{m-p}(\Sigma;\bZ)$ using \eqref{eqTZCup} and \eqref{diaDiffCoho}: 
\begin{align}
\la\cdot,\cdot\ra_\tor: \Ht^p(\Sigma;\bZ) \times \Ht^{m-p}(\Sigma;\bZ) \to \bT 
	&&	(t_\Sigma,t_\Sigma^\prime) \longmapsto (u_\Sigma \smile \jx t_\Sigma^\prime)[\Sigma], 
\end{align}
for any $u_\Sigma \in \H^{p-1}(\Sigma;\bT)$ such that $\be u_\Sigma = t_\Sigma$. Notice that this definition is well-posed on account of the properties of the cup product. 
Using $\la\cdot,\cdot\ra_\tor$, we can introduce the bi-homomorphism: 
\begin{align}
\sigma_\tor^\Sigma: \Ht^{k,m-k}(\Sigma;\bZ) \times \Ht^{k,m-k}(\Sigma;\bZ) \longrightarrow \bT, 
	&&	\big( (t_\Sigma, \tilde t_\Sigma), (t_\Sigma^\prime, \tilde t_\Sigma^\prime) \big) 
			\longrightarrow \la \tilde t_\Sigma, t_\Sigma^\prime \ra_\tor 
								- \la \tilde t_\Sigma^\prime, t_\Sigma \ra_\tor.
\end{align}
Since the pairing $\la\cdot,\cdot\ra_\tor$ is non-degenerate, also $\sigma_\tor^\Sigma$ is such. 
Furthermore, it is clearly antisymmetric, hence $(\Ht^{k,m-k}(\Sigma;\bZ),\sigma_\tor^\Sigma)$ 
is symplectic Abelian group. The isomorphism allows us to transfer 
$\sigma_\tor^\Sigma$ to $\Topt^k(M)$ by setting: 
\begin{equation}\label{eqToptSympl}
\sigma_\tor \doteq \sigma_\tor^\Sigma \circ (i_\Sigma^\ast \times i_\Sigma^\ast): 
	\Topt^k(M) \times \Topt^k(M) \longrightarrow \bT. 
\end{equation}
The argument that makes the definition in \eqref{eqTopfSympl} independent of the choice of Cauchy surface 
can be applied here too to show that also $\sigma_\tor$ does not depend on such choice. 
With the last equation, we have endowed $\Topt^k(M)$ with a natural symplectic structure, 
so that $(\Topt^k(M),\sigma_\tor)$ is a symplectic Abelian group. 
\vskip 1em

Similarly to the dynamical and the torsion-free topological sectors, 
also the torsion topological sector carries a natural duality isomorphism, 
compatible with the symplectic structure $\sigma_\tor$: 
\begin{align}\label{eqn:dualityTor}
\zeta_\tor: \Topt^k(M) \longrightarrow \Topt^{m-k}(M), && (t, \tilde t) \longmapsto (\tilde t, (-1)^{k(m-k)+1} t). 
\end{align}

\subsection{\label{secDecomposition}Symplectically orthogonal decomposition}
Recall that $M$ is an $m$-dimensional globally hyperbolic spacetime 
admitting a compact spacelike Cauchy surface $\Sigma$. 
So far the assumption of a compact Cauchy surface was just meant to simplify our presentation. 
In Theorem \ref{thmSplitting} this assumption will be used in a crucial way. 

We will now present a procedure to decompose orthogonally (however not canonically) 
the symplectic Abelian group $(\Conf^k(M;\bZ),\sigma)$ 
into the three symplectic Abelian groups $(\Dyn^k(M),\sigma_\Dyn)$, 
$(\Topf^k(M),\sigma_\free)$ and $(\Topt^k(M),\sigma_\tor)$ 
introduced in the previous sections. 
To achieve this result, we need to choose (non-canonical) splittings 
for the short exact sequences in \eqref{diaConf}. 
Before we prove that splittings of the desired type actually exist, 
let us illustrate the assumptions that ensure the compatibility 
of our decomposition with the relevant symplectic structures. 

\begin{lemma}\label{lemSplitting}
Consider a commutative diagram 
\begin{equation}
\xymatrix{
		&	0 \ar[d]	&	0 \ar[d]	&	0 \ar[d]														\\
0 \ar[r]	&	A_1 \ar[r]^-{i_1} \ar[d]_-{a_1}
					&	E_1 \ar[r]_-{p_1} \ar[d]_-{e_1}
								&	B_1 \ar[r] \ar[d]^-{b_1}	\ar@/_1pc/[l]_-{\pi_1}		&	0	\\
0 \ar[r]	&	A_2 \ar[r]^-{i_2} \ar[d]^-{a_2}
					&	E_2 \ar[r]^-{p_2} \ar[d]_-{e_2}
								&	B_2 \ar[r] \ar[d]^-{b_2}								&	0	\\
0 \ar[r]	&	A_3 \ar[r]_-{i_3} \ar[d] \ar@/^1pc/[u]^-{\alpha_2}
					&	E_3 \ar[r]_-{p_3} \ar[d]
								&	B_3 \ar[r] \ar[d] \ar[ul]^-\chi						&	0	\\
		&	0		&	0		&	0
}
\end{equation}
of Abelian groups whose rows and columns are short exact sequences 
and assume that $\pi_1$, $\alpha_2$, $\chi$ split the relevant short exact sequences, i.e. 
\begin{align}\label{eqSplitId}
p_1 \circ \pi_1 = \id_{B_1},	&&		a_2 \circ \alpha_2 = \id_{A_3},	&&	b_2 \circ p_2 \circ \chi = \id_{B_3}. 
\end{align}
Then 
\begin{equation}
I \doteq e_1 \circ \pi_1 + e_1 \circ i_1 + \chi + i_2 \circ \alpha_2: 
	B_1 \oplus A_1 \oplus B_3 \oplus A_3 \longrightarrow E_2 
\end{equation}
is an isomorphism of Abelian groups. 
\end{lemma}
\begin{proof}
Commutativity of the bottom right square and the last identity of \eqref{eqSplitId} entail that 
$\beta_2 \doteq p_2 \circ \chi: B_3 \to B_2$ splits the left column, 
while $\pi_3 \doteq e_2 \circ \chi: B_3 \to E_3$ splits the bottom row. 
Therefore, by exactness there exist unique homomorphisms $\beta_1: B_2 \to B_1, \iota_3: E_3 \to A_3$ such that 
\begin{align}\label{eqn:auxsplit1}
\beta_2 \circ b_2 + b_1 \circ \beta_1 = \id_{B_2},	&&	\pi_3 \circ p_3 + i_3 \circ \iota_3 = \id_{E_3}. 
\end{align}
Similarly, the first two identities of \eqref{eqSplitId} entail that there exist unique homomorphisms 
$\iota_1: E_1 \to A_1$ and $\alpha_1: A_2 \to A_1$ such that 
\begin{align}\label{eqn:auxsplit2}
\pi_1 \circ p_1 + i_1 \circ \iota_1 = \id_{E_1},	&&	\alpha_2 \circ a_2 + a_1 \circ \alpha_1 = \id_{A_2}. 
\end{align}
Combining \eqref{eqn:auxsplit1} and \eqref{eqn:auxsplit2}, 
we can also split the central column and row. In fact, introducing 
\begin{align}\label{eqCentralSplit}
\epsilon_2 \doteq \chi \circ p_3 + i_2 \circ \alpha_2 \circ \iota_3: E_3 \longrightarrow E_2, 
	&& \pi_2 \doteq \chi \circ b_2 + e_1 \circ \pi_1 \circ \beta_1: B_2 \longrightarrow E_2, 
\end{align}
it is easy to confirm that $e_2 \circ \epsilon_2 = \id_{E_3}$ and $p_2 \circ \pi_2 = \id_{B_2}$. 
In particular, by exactness there exist unique homomorphisms 
$\epsilon_1: E_2 \to E_1, \iota_2: E_2 \to A_2$ such that 
\begin{align}
\pi_2 \circ p_2 + i_2 \circ \iota_2 = \id_{E_2},	&&	\epsilon_2 \circ e_2 + e_1 \circ \epsilon_1 = \id_{E_2}. 
\end{align}
With these preparations, we obtain a candidate for the inverse of $I$: 
\begin{equation}
J \doteq \big( (p_1,\iota_1) \oplus (p_3,\iota_3) \big) \circ (\epsilon_1,e_2): 
	E_2 \longrightarrow B_1 \oplus A_1 \oplus B_3 \oplus A_3. 
\end{equation}
To confirm that $J$ is the inverse of $I$, observe that \eqref{eqCentralSplit} entails 
the identities $\epsilon_2 \circ \pi_3 = \chi$ and $\epsilon_2 \circ i_3 = i_2 \circ \alpha_2$, hence 
\begin{equation}
I = (e_1 + \epsilon_2) \circ \big( (\pi_1 + i_1) \oplus (\pi_3 + i_3) \big): 
	B_1 \oplus A_1 \oplus B_3 \oplus A_3 \longrightarrow E_2. 
\end{equation}
Using the splitting identities, we conclude that $I \circ J$ and $J \circ I$ are the appropriate identity morphisms. 
\end{proof}

In the specific case of \eqref{diaConf}, the splittings we are interested in are depicted in the diagram below: 
\begin{equation}\label{diaSplittings}
\xymatrix@C=2.5em{
		&	0 \ar[d]		&	0 \ar[d]		&	0\ar[d]														\\
0 \ar[r]	&	\frac{\H^{k-1,m-k-1}(M;\bR)}{\Hf^{k-1,m-k-1}(M;\bZ)}
				\ar[r]^-{\nux \times \nux} \ar[d]_-{\mux \times \mux}
						&	\TT^k(M;\bZ) \ar[r]_-{\dd_1} \ar[d]_-{\tr \times \tr} 
										&	\Dyn^k(M) 
												\ar[r] \ar[d]^-\subseteq \ar@/_1.5pc/[l]_-\eta 		&	0	\\
0 \ar[r]	&	\H^{k-1,m-k-1}(M;\bT) \ar[r]^-{\fl \times \fl} \ar[d]^-{\be \times \be}
						&	\Conf^k(M;\bZ) \ar[r]^-{\cu_1} \ar[d]_-{\ch \times \ch}
										&	\OmegaZ^k \cap \ast\,\OmegaZ^{m-k}(M)
												\ar[r] \ar[d]^-{([\,\cdot\,],[\ast^{-1}\,\cdot\,])}	&	0	\\
0 \ar[r]	&	\Topt^k(M) \ar[r]_-{\jx \times \jx} \ar[d] \ar@/^1pc/[u]^-\xi 
						&	\H^{k,m-k}(M;\bZ) \ar[r]_-{\qx \times \qx} \ar[d]
										&	\Hf^{k,m-k}(M;\bZ) \ar[r] \ar[d] \ar[ul]^-\chi 		&	0	\\
		&	0			&	0			&	0
}
\end{equation}
More explicitly, the splitting conditions read 
\begin{align}\label{eqSplittings}
(\qx \times \qx) \circ (\ch \times \ch) \circ \chix = \id_{\Hf^{k,m-k}(M;\bZ)}, 
	&&	(\beta \times \beta) \circ \xix = \id_{\Topt^k(M)},		&&	\dd_1 \circ \etax = \id_{\Dyn^k(M)}. 
\end{align}

\begin{lemma}\label{lemSymplSplit}
Let $M$ be an $m$-dimensional globally hyperbolic spacetime 
admitting a compact spacelike Cauchy surface $\Sigma$. 
Assuming splittings $\chix: \Hf^{k,m-k}(M;\bZ) \to \Conf^k(M;\bZ)$, $\etax: \Dyn^k(M) \to \TT^k(M;\bZ)$ 
and $\xix: \Topt^k(M) \to \H^{k-1,m-k-1}(M;\bT)$ according to \eqref{diaSplittings} and \eqref{eqSplittings}, 
the following identities are fulfilled: 
\begin{subequations}\label{eqn:symplidentities}
\begin{align}
\sigma \circ \Big( \big( (\tr \times \tr) \circ \etax \big) \times \big( (\tr \times \tr) \circ \etax \big) \Big) 
	&	= \sigma_\Dyn,	\\
\sigma \circ \Big( \big( (\fl \times \fl) \circ \xix \big) \times \big( (\fl \times \fl) \circ \xix \big) \Big) 
	&	= \sigma_\tor,	\\
\sigma \circ \big( (\tr \times \tr) \times (\fl \times \fl) \big)	&	= 0. 
\end{align}
Furthermore, for each 
$(u,\tilde u, z, \tilde z), (u^\prime,\tilde u^\prime, z^\prime, \tilde z^\prime) \in \Topf^k(M)$, one has 
\begin{equation}
\sigma \big( (\tr \times \tr) (\nux \times \nux) (u, \tilde u), \chix (z^\prime, \tilde z^\prime) \big) 
	+ \sigma \big( \chix (z, \tilde z), (\tr \times \tr) (\nux \times \nux) (u^\prime, \tilde u^\prime) \big) 
	= \sigma_\free \big( (u,\tilde u, z, \tilde z), (u^\prime,\tilde u^\prime, z^\prime, \tilde z^\prime) \big). 
\end{equation}
\end{subequations}
\end{lemma}
\begin{proof}
It is easier to show the desired identities using the equivalent definition of the relevant symplectic structures 
in terms of data on the Cauchy surface $\Sigma$. The proof relies uses \eqref{eqRingComp} extensively. 
Consider $\dd A, \dd A^\prime \in \Dyn^k(M)$ and denote the images of $(\tr \times \tr) \etax \dd A$ 
and $(\tr \times \tr) \etax \dd A^\prime$ along the isomorphism \eqref{eqCauchyIso} 
by $(\tr [A_\Sigma], \tr [\tilde A_\Sigma]) \in \dH^{k,m-k}(\Sigma;\bZ)$ 
and respectively by $(\tr [A^\prime_\Sigma], \tr [\tilde A^\prime_\Sigma]) \in \dH^{k,m-k}(\Sigma;\bZ)$. 
According to Section \ref{secObs}, 
to determine $\sigma \big( (\tr \times \tr) \etax \dd A, (\tr \times \tr) \etax \dd A^\prime \big) \in \bT$ 
we have to consider 
\begin{align}
\tr [\tilde A_\Sigma] \cdot \tr [A^\prime_\Sigma] - \tr [\tilde A^\prime_\Sigma] \cdot \tr [A_\Sigma] 
&	= \tr \big[ \tilde A_\Sigma \wedge \cu \tr [A^\prime_\Sigma] \big] - \tr \big[ \tilde A^\prime_\Sigma \wedge \cu \tr [A_\Sigma] \big]	\nonumber \\
&	= \tr [\tilde A_\Sigma \wedge \dd_\Sigma A^\prime_\Sigma] 
		- \tr [\tilde A^\prime_\Sigma \wedge \dd_\Sigma A_\Sigma]. 
\end{align}
Notice the use of \eqref{eqRingComp} to establish the first equality. 
Evaluation on the fundamental class $[\Sigma] \in \H_{m-1}(\Sigma)$ of $\Sigma$ 
concludes the proof of the first identity of \eqref{eqn:symplidentities} 
once one recalls also Section \ref{secDyn} and \eqref{eqSplittings}, which entail 
\begin{equation}
i_\Sigma^\ast \dd A = i_\Sigma^\ast \dd_1 \etax \dd A = \dd_1 ([A^\prime_\Sigma], [\tilde A^\prime_\Sigma]). 
\end{equation}
A similar argument proves the second identity of \eqref{eqn:symplidentities} too. 
\vskip .5em

To prove the third identity, consider $([A_\Sigma],[\tilde A_\Sigma]) 
\in \Omega^{k-1,m-k-1}(\Sigma) / \OmegaZ^{k-1,m-k-1}(\Sigma)$ 
and $(u_\Sigma, \tilde u_\Sigma) \in \H^{k-1,m-k-1}(\Sigma;\bT)$ and evaluate 
$\sigma^\Sigma \big( (\tr[A_\Sigma], \tr[\tilde A_\Sigma]), (\fl u_\Sigma, \fl \tilde u_\Sigma)\big)$. 
According to Section \ref{secObs}, this involves the following computation: 
\begin{align}\nonumber
\tr[\tilde A_\Sigma] \cdot \fl u_\Sigma - \fl \tilde u_\Sigma \cdot \tr[A_\Sigma] 
&	= \tr[\tilde A_\Sigma \wedge \cu \fl u_\Sigma] - \fl (\tilde u_\Sigma \smile \ch \tr[A_\Sigma]) \\
&	= \tr[\tilde A_\Sigma \wedge 0] - \fl (\tilde u_\Sigma \smile 0) = 0. 
\end{align}
Notice that we are again using \eqref{eqRingComp}. 
\vskip .5em

For the last identity of \eqref{eqn:symplidentities}, 
take $(u_\Sigma,\tilde u_\Sigma), (u_\Sigma^\prime,\tilde u_\Sigma^\prime) 
\in \H^{k-1,m-k-1}(\Sigma;\bR) / \Hf^{k-1,m-k-1}(\Sigma;\bZ)$ and 
$(z, \tilde z), (z^\prime, \tilde z^\prime) \in \Hf^{k,m-k}(M;\bZ)$ and introduce 
\begin{align}
(h_\Sigma, \tilde h_\Sigma) 
		\doteq i_\Sigma^\ast \chix (z, \tilde z) \in \dH^{k,m-k}(\Sigma;\bZ), 
	&&	(h_\Sigma^\prime, \tilde h_\Sigma^\prime) 
		\doteq i_\Sigma^\ast \chix (z^\prime, \tilde z^\prime) \in \dH^{k,m-k}(\Sigma;\bZ), 
\end{align}
where \eqref{eqCauchyIso} has been used. In view of Section \ref{secObs} computing
$\sigma^\Sigma \big( (\tr \nux u_\Sigma, \tr \nux \tilde u_\Sigma), 
(h_\Sigma^\prime, \tilde h_\Sigma^\prime) \big)$ involves the following calculation, 
based on \eqref{diaInitialData} and \eqref{eqRingComp}: 
\begin{align}\label{eqn:auxsplit3}
\tr \nux \tilde u_\Sigma \cdot h_\Sigma^\prime - \tilde h_\Sigma^\prime \cdot \tr \nux u_\Sigma \nonumber
&	= \fl \mux \tilde u_\Sigma \cdot h_\Sigma^\prime - \tilde h_\Sigma^\prime \cdot \fl \mux u_\Sigma	\\
&	= \fl ( \mux \tilde u_\Sigma \smile \ch h_\Sigma^\prime) 
		- (-1)^{k(m-k)} \fl (\mux u_\Sigma \smile \ch \tilde h_\Sigma^\prime). 
\end{align}
On account of \eqref{eqSplittings}, we observe that the element
$(\qx \times \qx) (\ch \times \ch) (h_\Sigma^\prime, \tilde h_\Sigma^\prime)$ 
is precisely the restriction $(z_\Sigma^\prime, \tilde z_\Sigma^\prime)$ to $\Sigma$ 
of $(z^\prime, \tilde z^\prime)$. 
Therefore, evaluating \eqref{eqn:auxsplit3} on the fundamental class 
$[\Sigma] \in \H_{m-1}(\Sigma)$ of $\Sigma$ and recalling \eqref{eqFreePairing}, we obtain 
\begin{equation}
\sigma^\Sigma \big( (\tr \nux u_\Sigma, \tr \nux \tilde u_\Sigma), 
		(h_\Sigma^\prime, \tilde h_\Sigma^\prime) \big) 
	= \la \tilde u_\Sigma, z_\Sigma^\prime \ra_\free 
		- (-1)^{k(m-k)} \la u_\Sigma, \tilde z_\Sigma^\prime \ra_\free. 
\end{equation}
A similar argument shows that 
\begin{equation}
\sigma^\Sigma \big( (h_\Sigma, \tilde h_\Sigma) \big), 
		(\tr \nux u_\Sigma^\prime, \tr \nux \tilde u_\Sigma^\prime) \big) 
	= - \la \tilde u_\Sigma^\prime, z_\Sigma \ra_\free 
		+ (-1)^{k(m-k)} \la u_\Sigma^\prime, \tilde z_\Sigma \ra_\free. 
\end{equation}
We conclude combining the last two equations and recalling Section \ref{secTopf}. 
\end{proof}

Lemma \ref{lemSplitting} entails that the one defined below is an isomorphism of Abelian groups: 
\begin{align}\label{eqDecomposition}
\Dyn^k(M) \oplus \Topf^k(M) \oplus \Topt^k(M)		&	\overset{\simeq}{\longrightarrow} \Conf^k(M;\bZ),	\\
\big( \dd A, (u, \tilde u, z, \tilde z), (t, \tilde t) \big)	&	\longmapsto (\tr \times \tr) \big( \etax \dd A 
	+ (\nux \times \nux) (u, \tilde u) \big) + \chix (z, \tilde z) + (\fl \times \fl) \xix (t, \tilde t). \nonumber
\end{align}
Furthermore, if we assume that the splittings fulfil the compatibility conditions  
\begin{align}\label{eqSymplComp}
\sigma \circ (\chix \times \chix) = 0, 
	&&	\sigma \circ \Big( \big( (\fl \times \fl) \circ \xix \big) \times \chix \Big) = 0, 
	&&	\sigma \circ \Big( \big( (\tr \times \tr) \circ \etax \big) \times \chix \Big) = 0 
\end{align}
with respect to the symplectic structure $\sigma$ on $\Conf^k(M;\bZ)$, 
recalling Lemma \ref{lemSymplSplit}, we conclude that 
\eqref{eqDecomposition} is also an isomorphism of symplectic Abelian groups. 
The symplectic structure on the source is obtained combining those of each summand, 
cf.\ \eqref{eqDynSympl}, \eqref{eqTopfSympl} and \eqref{eqToptSympl}. 
In fact, Lemma \ref{lemSymplSplit} and \eqref{eqSymplComp} entail the following identity: 
\begin{align}
\sigma & \left( 
\begin{array}{l}
(\tr \times \tr) \etax \dd A + (\tr \times \tr) (\nux \times \nux) (u, \tilde u) 
	+ \chix (z, \tilde z) + (\fl \times \fl) \xix (t, \tilde t), \\
(\tr \times \tr) \etax \dd A^\prime + (\tr \times \tr) (\nux \times \nux) (u^\prime, \tilde u^\prime) 
	+ \chix (z^\prime, \tilde z^\prime) + (\fl \times \fl) \xix (t^\prime, \tilde t^\prime) 
\end{array}
\right) \\
&	= \sigma_\Dyn (\dd A, \dd A^\prime) 
	+ \sigma_\free \big( (u, \tilde u, z, \tilde z), (u^\prime, \tilde u^\prime, z^\prime, \tilde z^\prime) \big) 
	+ \sigma_\tor \big( (t, \tilde t), (t^\prime, \tilde t^\prime) \big). 
\end{align}

\begin{theorem}\label{thmSplitting}
Let $M$ be an $m$-dimensional globally hyperbolic spacetime 
admitting a compact spacelike Cauchy surface $\Sigma$. 
Then there exist splittings as per \eqref{diaSplittings} and \eqref{eqSplittings} 
fulfilling \eqref{eqSymplComp} and compatible with dualities, 
cf.\ \eqref{eqn:duality}, \eqref{eqn:dualityFree}, \eqref{eqn:dualityTor}. 
\end{theorem}
\begin{proof}
In practice, it is easier to work using the equivalent description of $\Conf^k(M;\bZ)$ in terms of initial data 
provided by the restriction along the embedding $i_\Sigma: \Sigma \to M$ 
of the spacelike Cauchy surface $\Sigma$ into the spacetime $M$. 
In particular, we will construct the splittings with reference to \eqref{diaInitialData}. 
The actual statement of the theorem is then obtained via the isomorphism 
relating \eqref{diaConf} to \eqref{diaInitialData}, see \eqref{eqCauchyIso}. 
\vskip .5em

We start constructing $\chix: \Hf^{k,m-k}(\Sigma;\bZ) \to \dH^{k,m-k}(\Sigma;\bZ)$ 
such that $(\qx \times \qx) \circ (\ch \times \ch) \circ \chix = \id_{\Hf^{k,m-k}(\Sigma;\bZ)}$ 
and $\sigma^\Sigma \circ (\chix \times \chix) = 0$. 
By definition $\Hf^p(\Sigma;\bZ)$ is a free Abelian group. In particular, we can choose bases 
\begin{align}
\{z_i:\, i = 1, \ldots, n\} \subseteq \Hf^k(\Sigma;\bZ), 
	&&	\{\tilde z_{\tilde i}:\, \tilde i = 1, \ldots, \tilde n\} \subseteq \Hf^{m-k}(\Sigma;\bZ). 
\end{align} 
Since both $\qx: \H^p(\Sigma;\bZ) \to \Hf^p(\Sigma;\bZ)$ and $\ch: \dH^p(\Sigma;\bZ) \to \H^p(\Sigma;\bZ)$ 
are surjective, we can choose 
\begin{subequations}
\begin{align}
\{h_i:\, i = 1, \ldots, n\} \subseteq \dH^k(\Sigma;\bZ), 
	&&	\{\tilde h^\prime_{\tilde i}:\, \tilde i = 1, \ldots, \tilde n\} \subseteq \dH^{m-k}(\Sigma;\bZ) 
\end{align}
such that 
\begin{align}
\qx \ch h_i = z_i,	&&	\qx \ch \tilde h^\prime_{\tilde i} = \tilde z_{\tilde i}. 
\end{align}
\end{subequations}
By evaluation on the fundamental class $[\Sigma]$ of $\Sigma$, we introduce a set of real numbers 
\begin{align}
\{c_{\tilde i i}:\, \tilde i = 1, \ldots, \tilde n,\, i = 1, \ldots, n\} \subseteq \bR: 
	&&	c_{\tilde i i} \mod \bZ = (\tilde h^\prime_{\tilde i} \cdot h_i)[\Sigma]. 
\end{align}
Consider the non-degenerate pairing 
\begin{align}\label{eqn:realcohopairing}
\H^{m-k-1}(\Sigma;\bR) \times \H^k(\Sigma;\bR) \longrightarrow \bR, 
	&& (\tilde r, r) \longmapsto (\tilde r \smile r)[\Sigma] 
\end{align}
for cohomology with real coefficients on $\Sigma$. 
Since $\Hf^k(\Sigma;\bZ)$ is a lattice in $\H^k(\Sigma;\bR)$, $\{z_i\}$ is a basis of $\H^k(\Sigma;\bR)$ too. 
Then we can select its dual basis via the non-degenerate pairing displayed above: 
\begin{align}
\{\tilde r_i:\, i = 1, \ldots, n\} \subseteq \H^{m-k-1}(\Sigma;\bR): 
	&&	(\tilde r_i \smile z_j)[\Sigma] = \delta_{ij}. 
\end{align}
We define 
\begin{subequations}
\begin{align}
\{\tilde s_{\tilde i}:\, \tilde i = 1, \ldots, \tilde n\} \subseteq \H^{m-k-1}(\Sigma;\bR), 
	&&	\{\tilde u_{\tilde i}:\, \tilde i = 1, \ldots, \tilde n\} 
		\subseteq \frac{\H^{m-k-1}(\Sigma;\bR)}{\Hf^{m-k-1}(\Sigma;\bZ)} 
\end{align}
according to 
\begin{align}
\tilde s_{\tilde i} \doteq \sum_{i=1}^n c_{\tilde i i}\, \tilde r_i, 
	&&	\tilde u_{\tilde i} \doteq \tilde s_{\tilde i} \mod \Hf^{m-k-1}(\Sigma;\bZ). 
\end{align}
\end{subequations}
By construction, one finds 
\begin{equation}
(\tilde s_{\tilde i} \smile z_i)[\Sigma] = \sum_{j=1}^n c_{\tilde i j} (\tilde r_j \smile z_i)[\Sigma] = c_{\tilde i i}. 
\end{equation}
Therefore, setting also
\begin{align}
\{\tilde h_{\tilde i} \doteq \tilde h^\prime_{\tilde i} - \fl \mux \tilde u_{\tilde i}:\, \tilde i = 1, \ldots, \tilde n\} 
	\subseteq \dH^{k,m-k}(\Sigma;\bZ), 
\end{align}
we get 
\begin{align}
(\tilde h_{\tilde i} \cdot h_i)[\Sigma] 
&	= (\tilde h^\prime_{\tilde i} \cdot h_i)[\Sigma] - (\mux \tilde u_{\tilde i} \smile \ch h_i)[\Sigma]	\nonumber \\
&	= c_{\tilde i i} - (\tilde s_{\tilde i} \smile \qx \ch h_i)[\Sigma] \mod \bZ							\nonumber \\
&	= c_{\tilde i i} - (\tilde s_{\tilde i} \smile z_i)[\Sigma] \mod \bZ = 0. 
\end{align}
Hence, the formula 
\begin{align}
\chix: \Hf^{k,m-k}(\Sigma;\bZ) \longrightarrow \dH^{k,m-k}(\Sigma:\bZ), 
	&&	(z_i,0) \longmapsto (h_i,0),	&&	(0,\tilde z_{\tilde i}) \longmapsto (0,\tilde h_{\tilde i}) 
\end{align}
uniquely specifies the sought homomorphism on the basis 
$\{(z_i,0), (0, \tilde z_{\tilde i}): i = 1, \ldots, n, \tilde i = 1, \ldots, \tilde n\}$ of $\Hf^{k,m-k}(\Sigma;\bZ)$. 
The splitting $\chi$ can be made compatible with the dualities $\zeta_\free$ in \eqref{eqn:dualityFree} 
and $\zeta$ in \eqref{eqn:duality}. In fact, it is sufficient to consider also 
\begin{align}
\tilde \chix: \Hf^{m-k,k}(\Sigma;\bZ) \longrightarrow \dH^{m-k,k}(\Sigma:\bZ), 
	&&	(\tilde z_{\tilde i}, 0) \longmapsto (\tilde h_{\tilde i}, 0),		&&	(0, z_i) \longmapsto (0, h_i) 
\end{align}
to conclude that $\zeta \circ ((\tr \nux \times \tr \nux) \times \chix) 
= ((\tr \nux \times \tr \nux) \times \tilde \chix) \circ \zeta_\free$. 
\vskip .5em

As a second step, we focus on the construction of 
\begin{equation}
\eta: \dd_\Sigma \Omega^{k-1,m-k-1}(\Sigma) \longmapsto 
	\frac{\Omega^{k-1,m-k-1}(\Sigma)}{\OmegaZ^{k-1,m-k-1}(\Sigma)} 
\end{equation} 
such that $(\dd_\Sigma \times \dd_\Sigma) \circ \eta = \id_{\dd_\Sigma \Omega^{k-1,m-k-1}(\Sigma)}$ 
and $\sigma^\Sigma \circ \Big( \big( (\tr \times \tr) \circ \etax \big) \times \chix \Big) = 0$. 
First of all, we observe that a splitting $\etax^\prime = \etax^\prime_1 \times \etax^\prime_2$ exists. 
In fact, on account of \cite[Section A.1]{BSS14}, we obtain $\etax^\prime_1$ and $\etax^\prime_2$ 
such that $\dd_\Sigma \circ \etax^\prime_1 = \id_{\dd_\Sigma \Omega^{k-1}(\Sigma)}$ 
and $\dd_\Sigma \circ \etax^\prime_2 = \id_{\dd_\Sigma \Omega^{m-k-1}(\Sigma)}$. 
Our goal is to define $\eta$ as a suitable modification of $\eta^\prime$. 
For this purpose, we observe that $\H^p(\Sigma;\bR) / \Hf^p(\Sigma;\bZ)$ 
is the Pontryagin dual of $\Hf^{m-p-1}(\Sigma;\bZ)$, cf.\ \cite[Rem.\ 5.7]{BMath}. 
In particular, recalling the definition of $\sigma_\free^\Sigma$ in \eqref{eqTopfSympl}, we observe that 
\begin{align}
\frac{\H^{k-1,m-k-1}(\Sigma;\bR)}{\Hf^{k-1,m-k-1}(\Sigma;\bZ)} 
		\longrightarrow \Hf^{k,m-k}(\Sigma;\bZ)^\star,	
	&&	(u,\tilde u) \longmapsto \sigma_\free^\Sigma 
		\longmapsto \sigma_\free^\Sigma \big((u,\tilde u, 0, 0), \cdot \big) 
\end{align}
provides the Pontryagin duality isomorphism. 
This observation allows us to define 
\begin{subequations}\label{eqDeltaEta}
\begin{equation}
\Delta_\eta: \dd_\Sigma \Omega^{k-1,m-k-1}(\Sigma) \longmapsto 
	\frac{\H^{k-1,m-k-1}(\Sigma;\bR)}{\Hf^{k-1,m-k-1}(\Sigma;\bZ)} 
\end{equation}
by setting 
\begin{equation}
\sigma^\Sigma_\free \Big( \big( \Delta_\eta(\dd_\Sigma A, \dd_\Sigma \tilde A), 0, 0 \big), (0,0,z,\tilde z) \Big) 
\doteq \sigma^\Sigma \big( (\tr \times \tr) \etax^\prime (\dd_\Sigma A, \dd_\Sigma \tilde A) , \chi(z,\tilde z) \big) 
\end{equation}
\end{subequations}
for each $(\dd_\Sigma A, \dd_\Sigma \tilde A) \in \dd_\Sigma \Omega^{k-1,m-k-1}(\Sigma)$ 
and each $(z,\tilde z) \in \Hf^{k,m-k}(\Sigma;\bZ)$. 
In fact, for each $(\dd_\Sigma A, \dd_\Sigma \tilde A)$, the right-hand side yields a group character 
on $\Hf^{k,m-k}(\Sigma;\bZ)$; hence, $\Delta_\eta(\dd_\Sigma A, \dd_\Sigma \tilde A)$ 
fulfilling the defining condition displayed above exists and is unique. 
Notice that $\Delta_\eta$ is actually the Cartesian product of 
two homomorphisms exactly as $\eta^\prime$ due to the fact that 
its defining equation \eqref{eqDeltaEta} does not mix components. Introducing: 
\begin{equation}
\etax \doteq \etax^\prime - (\nu \times \nu) \circ \Delta_\eta: 
	\dd_\Sigma \Omega^{k-1,m-k-1}(\Sigma) \longrightarrow 
		\frac{\Omega^{k-1,m-k-1}(\Sigma)}{\OmegaZ^{k-1,m-k-1}(\Sigma)} 
\end{equation}
and recalling also the last equation of Lemma \ref{lemSymplSplit}, 
we find that $\etax$ fulfils the desired requirement: 
\begin{align}\label{eqEtaChiOrthogonal}
\sigma^\Sigma & \Big(  (\tr \times \tr) \etax (\dd A, \dd \tilde A), \chix (z,\tilde z) \Big) \nonumber \\
&	= \sigma^\Sigma \Big(  (\tr \times \tr) \etax^\prime (\dd A, \dd \tilde A), \chix (z,\tilde z) \Big) 
	- \sigma^\Sigma \Big(  (\tr \times \tr) (\nu \times \nu) \Delta_\eta (\dd A, \dd \tilde A), \chix (z,\tilde z) \Big) \nonumber \\
&	= \sigma^\Sigma_\free \Big( \big( \Delta_\eta(\dd A, \dd \tilde A), 0, 0 \big), (0,0,z,\tilde z) \Big) 
	- \sigma^\Sigma_\free \Big( \big( \Delta_\eta(\dd A, \dd \tilde A), 0, 0 \big), (0,0,z,\tilde z) \Big) = 0. 
\end{align}
Once again $\eta$ is the Cartesian product of two morphisms, namely its first component $\etax_1$, 
which does not depend on the second argument of $\etax$, and its second component $\etax_2$, 
which is instead independent of the first argument. 
It is now easy to confirm the compatibility with the dualities 
$\zeta_\Dyn$ of \eqref{eqn:dualityDyn} and $\zeta$ of \eqref{eqn:duality}
by introducing a second splitting whose components are obtained 
flipping the components of the splitting constructed above: 
\begin{align}
\tilde \etax: \dd_\Sigma \Omega^{m-k-1,k-1}(\Sigma) \longrightarrow 
	\frac{\Omega^{m-k-1,k-1}(\Sigma)}{\OmegaZ^{m-k-1,k-1}(\Sigma)} 
	&&	(\dd_\Sigma A, \dd_\Sigma \tilde A) \longmapsto (\etax_2(\dd_\Sigma A), \etax_1(\dd_\Sigma \tilde A)). 
\end{align}
As a consequence, we find $\zeta \circ (\tr \times \tr) \circ \etax = (\tr \times \tr) \circ \tilde \etax \circ \zeta_\Dyn$, 
which is the desired compatibility. 
\vskip .5em

The last step consists in providing $\xix: \Ht^{k,m-k}(\Sigma;\bZ) \to \H^{k-1,m-k-1}(\Sigma;\bT)$ 
such that $(\be \times \be) \circ \xix = \id_{\Ht^{k,m-k}(\Sigma;\bZ)}$ 
and $\sigma^\Sigma \circ \Big( \big( (\fl \times \fl) \circ \xix \big) \times \chix \Big) = 0$. 
To start with, note that $\H^{p-1}(\Sigma;\bR) / \Hf^{p-1}(\Sigma;\bZ)$ is a divisible group, 
so that there exists a splitting $\xix^\prime = \xix_1^\prime \times \xix_2^\prime:
\Ht^{k,m-k}(\Sigma;\bZ) \to \H^{k-1,m-k-1}(\Sigma;\bT)$. 
It only remains to modify $\xix^\prime$ in order to obtain the desired $\xix$. 
An argument similar to the one we used to define $\Delta_\eta$, cf.\ \eqref{eqDeltaEta}, allows us to introduce 
\begin{subequations}
\begin{equation}
\Delta_\xi: \Ht^{k,m-k}(\Sigma;\bZ) \longrightarrow 
	\frac{\H^{k-1,m-k-1}(\Sigma;\bR)}{\Hf^{k-1,m-k-1}(\Sigma;\bZ)} 
\end{equation}
by setting 
\begin{equation}
\sigma_\free^\Sigma \Big( \big( \Delta_\xi (t, \tilde t), 0, 0 \big), (0, 0, z, \tilde z) \Big) 
	\doteq \sigma^\Sigma \big( (\fl \times \fl) \xix^\prime (t, \tilde t), \chix (z, \tilde z) \big) 
\end{equation}
\end{subequations}
for each $(t, \tilde t) \in \Ht^{k,m-k}(\Sigma;\bZ)$ and each $(z, \tilde z) \in \Hf^{k,m-k}(\Sigma;\bZ)$. 
Now consider 
\begin{equation}
\xix \doteq \xix^\prime - (\mux \times \mux) \circ \Delta_\xi: \Ht^{k,m-k}(\Sigma;\bZ) 
	\longrightarrow \H^{k-1,m-k-1}(\Sigma;\bT). 
\end{equation}
Repeating the calculation in \eqref{eqEtaChiOrthogonal}, one can confirm that $\xix$ fulfils the desired property. 
By the same argument valid for $\Delta_\eta$, $\Delta_\xi$ is a Cartesian product of two morphisms 
and so is $\xix$. Introducing a new splitting with flipped components 
\begin{align}
\tilde \xix: \Ht^{m-k,k}(\Sigma;\bZ) \longrightarrow \H^{m-k-1,k-1}(\Sigma;\bT), 
&& (t, \tilde t) \longmapsto (\xix_2 t, \xix_1 \tilde t), 
\end{align}
one can confirm that $\zeta \circ (\iota \times \iota) \circ \xix = (\iota \times \iota) \circ \tilde \xix \circ \zeta_\tor$, 
which expresses the compatibility between the splittings $\xix$ and $\tilde \xix$ 
and the dualities $\zeta$ of \eqref{eqn:duality} and $\zeta_\tor$ of \eqref{eqn:dualityTor}. 
\end{proof}

\begin{remark}\label{rem:duality_m=2k}
While for generic $(m,k)$ the question does not make sense, 
for $m=2k$ one would like to find a splitting that is self-compatible with duality 
(Theorem \ref{thmSplitting} only guarantees the existence of a second splitting 
that agrees with the first one after duality). 
This question can be answered by finding splittings whose components are two copies of the same morphism. 
To illustrate how to achieve this result, let us find a suitable splitting of the form 
\begin{align}
\chix = \chix_1 \times \chix_2: \Hf^{k,k}(\Sigma;\bZ) \longrightarrow \dH^{k,k}(\Sigma;\bZ)
\end{align}
with $\chix_1 = \chix_2$. Adopting the notation of the proof of Theorem \ref{thmSplitting}, we consider dual bases 
\begin{align}
\{z_i:\, i = 1, \ldots, n\} \subseteq \Hf^k(\Sigma;\bZ), 
&& \{\tilde r_i:\, i = 1, \ldots, n\} \subseteq \H^{k-1}(\Sigma;\bR) 
\end{align}
and we choose arbitrarily 
\begin{align}
\{h_i^\prime:\, i = 1, \ldots, n\} \subseteq \dH^k(\Sigma;\bZ) && \mbox{such that} && \qx \ch h_i^\prime = z_i. 
\end{align}
Selecting a collection of real numbers 
\begin{align}
\{c_{ij} \in [0,1):\, i,j=1, \ldots, n\} && \mbox{such that} 
&& c_{ij} \mod \bZ = (h^\prime_i \cdot h^\prime_j)[\Sigma], 
\end{align}
one concludes that $c_{ij} = (-1)^{k^2} c_{ji}$, hence the set 
\begin{equation}
\Big\{ h_i \doteq h^\prime_i - \fl \mux \sum_{j=1}^n \tfrac{1}{2}\, c_{ij}\, \tilde r_j:\, i=1,\ldots,n \Big\} 
\end{equation}
satisfies the condition $(h_i \cdot h_j)[\Sigma] = 0$. Therefore we obtain the desired splitting specifying 
\begin{align}
\chix_1 = \chix_2: \Hf^k(\Sigma;\bZ) \longrightarrow \dH^k(\Sigma;\bZ), && z_i \longmapsto h_i. 
\end{align}
Similar conclusions follow for the other relevant splittings. 
In particular, one obtains a symplectically orthogonal decomposition also for the self-dual theory, 
which has been investigated in \cite[Sect.\ 7]{BPhys}. 
\end{remark}

\section{Quantization and states}\label{Sec:quantstates}
The goal of the present section is to construct quantized C$^\ast$-algebras of observables for the symplectic Abelian group $\Conf^k(M)$, the dynamical sector $\Dyn^k(M)$, the torsion-free topological sector $\Topf^k(M)$ and the torsion topological sector $\Topt^k(M)$. With Corollary \ref{isomorphism_dynamical_algebras} we will show that the decomposition of Theorem \ref{thmSplitting} has a quantum counterpart in terms of an appropriate tensor product of C$^\ast$-algebras. This allows us to define a state on the quantized C$^\ast$-algebra associated to $\Conf^k(M)$ by defining states on the C$^\ast$-algebras associated to each sector, cf.\ Proposition \ref{prp:stateDyn}, Proposition \ref{prp:statefree} and Proposition \ref{prp:statetor}. In particular, for the dynamical sector we construct a Hadamard state (this feature is not of interest for the other sectors).
In terms of induced GNS representations, one obtains Hilbert spaces equipped with isomorphisms 
implementing the duality $\zeta: \Conf^k(M;\bZ) \to \Conf^{m-k}(M;\bZ)$ of \eqref{eqn:duality}. 
In particular, for $m=2k$, these isomorphisms arise as unitary transformations. 
In Section \ref{sec:states_torsion_free} we will study in detail the GNS representation induced by the state on the torsion-free topological sector $\Topf^k(M)$ and we will make further comments on it in Remark \ref{rem:statefree}. 
The main results of this section are summarized in Theorem \ref{thmStates}.

\subsection{\label{sec:quantization}Quantization}
We quantize the symplectic Abelian group $\Conf^k(M;\bZ)$ implementing canonical commutation relations of Weyl type, thus obtaining a C$^\ast$-algebra of observables. In view of the symplectically orthogonal decomposition constructed in the previous section, we obtain an analogous factorization at the level of C$^\ast$-algebras. The analysis that we are going to present applies to any symplectically orthogonal decomposition. Although we are interested in applying it to $\Conf^k(M;\bZ)$, $\Dyn^k(M)$, $\Topf^k(M)$ and $\Topt^k(M)$, it is more convenient to work in the general setting, applying it to the case in hand only at the end. 
Let $G_1, G_2$ be Abelian groups equipped with a presymplectic form 
\begin{align}
\sigma_i:G_i\times G_i \longrightarrow \bT, && i=1,2,
\end{align}
i.e.\ an antisymmetric bi-homomorphism. 
For $i=1,2$ one forms the unital $*$-algebra $\mathcal{A}(G_i)$ generated by the symbols $\left\{ W (g_i), \, g_i\in G_i\right\}$ and subject to the defining relations 
\begin{align}\label{defining_relations}
W(g_i)^\ast = W(-g_i), && W(g_i)\, W(g^\prime_i) = \exp(2 \pi i\, \sigma_i(g_i,g^\prime_i))\, W(g_i+g^\prime_i).
\end{align}
Each of these algebras can be equipped with the following norm, defined in \cite{MAN73}:
\begin{align}\label{eqNorm1}
\|\cdot\|_1:\mathcal{A}(G_i)\longrightarrow\bR, && \Big\|\sum\limits_{j=1}^N\alpha_jW(g_j)\Big\|_1\doteq\sum_{j=1}^N|\alpha_j|,
\end{align}
where we consider arbitrary (but finite) linear combinations of the generators of $\mathcal{A}(G_i)$. Upon completion, we obtain Banach $*$-algebras
\begin{align}
\mB(G_i)\doteq\overline{(\mathcal{A}(G_i),\|\cdot\|_1)},\qquad i=1,2.
\end{align}
On the other hand, taking $(G_1\oplus G_2,\sigma)$ as our starting point, where we define
\begin{equation}\label{orthogonality}
\sigma((g_1,g_2),(g^\prime_1,g^\prime_2))\doteq\sigma_1(g_1,g^\prime_1)+\sigma_2(g_2,g^\prime_2),
\end{equation}
for all $(g_1,g_2),(g^\prime_1,g^\prime_2)\in G_1\oplus G_2$,
the same procedure can be repeated, resulting first in the $*$-algebra $\mathcal{A}(G_1\oplus G_2)$ and then in the Banach $*$-algebra $\mB(G_1\oplus G_2)$. The latter comes together with two canonical homomorphisms of Banach $*$-algebras 
\begin{equation}
\iota_i:\mB(G_i)\longrightarrow\mB(G_1\oplus G_2),\qquad i=1,2,
\end{equation}
which are completely specified by their action on generators, namely $\iota_1(W(g_1))\doteq W(g_1,0)$ for all $g_1 \in G_1$ and similarly for $\iota_2(W(g_2))\doteq W(0,g_2)$ for all $g_2 \in G_2$. Furthermore, we can consider the algebraic tensor product $\mB(G_1)\otimes\mB(G_2)$. 
This is a $\ast$-algebra with respect to the product and the involution that are defined componentwise 
by the counterparts on each factor. We equip $\mB(G_1)\otimes\mB(G_2)$ with the norm 
\begin{subequations}\label{tensor_product_norm}
\begin{align}
\|a\|_{\hat\otimes}=\inf\left(\sum\limits_{k=1}^N\|a_{k,1}\|_1\,\|a_{k,2}\|_1\right), && a\in\mB(G_1)\otimes\mB(G_2).
\end{align}
The infimum is taken over all possible presentations of $a$ as 
\begin{equation}
a=\sum_{k=1}^{N}a_{1,k}\otimes a_{2,k}, 
\end{equation}
\end{subequations}
with $a_{i,k}\in\mB(G_i)$. The completion of $\mB(G_1)\otimes\mB(G_2)$ with respect to \eqref{tensor_product_norm} leads to a Banach $*$-algebra \cite{Guichardet} denoted by 
\begin{equation}
\mB(G_1)\hat{\otimes}\mB(G_2). 
\end{equation}
Since \eqref{defining_relations} and \eqref{orthogonality} entail that $\iota_1(a_1)\, \iota_2(a_2) = \iota_2(a_2)\, \iota_1(a_1)$  for each $a_1\in\mB(G_1), a_2\in\mB(G_2)$, recalling the universal property of $\hat \otimes$, cf.\ \cite{Guichardet}, we obtain a Banach $\ast$-algebra morphism 
\begin{subequations}\label{continuous_morphism}
\begin{equation}
I:\mB(G_1)\hat{\otimes}\mB(G_2)\longrightarrow\mB(G_1\oplus G_2),
\end{equation}
uniquely specified by 
\begin{equation}
I(a_1 \otimes a_2) = \iota_1(a_1)\, \iota_2(a_2)
\end{equation}
for $a_i \in \mB(G_i)$. 
\end{subequations}
Our goal is to show that $I$ is an isomorphism of Banach $\ast$-algebras. It suffices to exhibit its inverse. In fact, consider the $\ast$-homomorphism
\begin{subequations}\label{inverse}
\begin{equation}
J: \mA(G_1 \oplus G_2) \longrightarrow \mB(G_1) \otimes \mB(G_2),
\end{equation}
defined on generators by 
\begin{equation}
J(W(g_1,g_2))=W(g_1)\otimes W(g_2),
\end{equation}
for $g_i \in G_i$. 
\end{subequations}
From \eqref{eqNorm1} and \eqref{tensor_product_norm}, one obtains the inequality 
\begin{equation}
\|J(a)\|_{\hat\otimes} \leq \|a\|_1 
\end{equation}
for all $a \in \mathcal{A}(G_1\oplus G_2)$, which entails that $J$ can be uniquely extended to a Banach $\ast$-algebra morphism after taking the completions on codomain and domain. With a slight abuse of notation, we denote this extension by 
\begin{equation}
J: \mB(G_1\oplus G_2) \longrightarrow \mB(G_1)\hat{\otimes}\mB(G_2).
\end{equation}
A direct inspection of the definitions of $I$ and $J$ unveils that 
\begin{align}
I\circ J= \textrm{id}_{\mB(G_1\oplus G_2)}, && J\circ I=\textrm{id}_{\mB(G_1)\hat{\otimes}\mB(G_2)},
\end{align}
which is tantamount to saying that $I$ is an isomorphism of Banach $\ast$-algebras. In other words, $\mB(G_1\otimes G_2)$ is isomorphic to $\mB(G_1)\hat{\otimes}\mB(G_2)$. To conclude our analysis we need to move to the level of C$^*$-algebras. To this end we recall that, to each unital Banach $*$-algebra, one can assign functorially its canonical enveloping C$^*$-algebra, see \cite[Sect.\ 2.7]{Dixmier}. This functor, which will be denoted by $\mC^\ast$, is the left-adjoint of the forgetful functor from unital C$^\ast$-algebras to unital Banach $\ast$-algebras. We consider the C$^*$-algebras $\mC^*(\mB(G_i))$, $i=1,2$, $\mC^*(\mB(G_1\oplus G_2))$ and $\mC^*(\mB(G_1)\hat{\otimes}\mB(G_2))$. Having already established that $\mB(G_1\oplus G_2)$ is isomorphic via \eqref{inverse} to $\mB(G_1)\hat{\otimes}\mB(G_2)$, by functoriality it follows that 
\begin{equation}
\mC^*(\mB(G_1\oplus G_2))\simeq \mC^*(\mB(G_1)\hat{\otimes}\mB(G_2)).
\end{equation}
Following \cite{Guichardet}, we introduce a C$^\ast$-norm $\|\cdot\|_{\check\otimes}$ on the algebraic tensor product $C_1 \otimes C_2$ of two C$^\ast$-algebras $C_1, C_2$ as the least upper bound of all C$^\ast$-subcross seminorms. The C$^\ast$-algebra $C_1 \check\otimes C_2$ obtained by completion with respect to $\|\cdot\|_{\check\otimes}$ is characterized by the following universal property: If $\psi_i: C_1 \to C_3$ ($i=1,2$) are two commuting morphisms of unital C$^*$-algebras from $C_i$ into $C_3$, then there exists a unique C$^*$-algebra morphism $\Psi:C_1\check{\otimes} C_2\to C_3$ such that $\Psi(c_1\otimes c_2) = \psi_1(c_1)\, \psi_2(c_2)$ for all $c_1\in C_1$ and for all $c_2\in C_2$. The following property relates $\hat\otimes$ and $\check\otimes$ via $\mC^\ast$ \cite{Guichardet}: given two unital Banach $\ast$-algebras, the enveloping C$^\ast$-algebra of their $\hat\otimes$-tensor product is naturally isomorphic to the $\check\otimes$-tensor product of their enveloping C$^\ast$-algebras. Therefore we obtain
\begin{equation}
\mC^*(\mB(G_1)\hat{\otimes}\mB(G_2)) \simeq \mC^*(\mB(G_1))\check{\otimes} \mC^*(\mB(G_2)).
\end{equation}
Summing up, we have the following:

\begin{proposition}\label{C*-isomorphism}
	Let $(G_1,\sigma_1)$ and $(G_2,\sigma_2)$ be two presymplectic Abelian groups. Then there is a canonical isomorphism of C$^\ast$-algebras:
	\begin{equation}
	\mC^*(\mB(G_1\oplus G_2))\simeq \mC^*(\mB(G_1))\check{\otimes} \mC^*(\mB(G_2)).
	\end{equation}
\end{proposition}

\begin{remark}
	We observe that, as a consequence of \cite[Prop.\ 2.7.1]{Dixmier}, for a (pre)symplectic Abelian group $G$, the C$^*$-enveloping algebra associated to $\mB(G)$ is isomorphic to the C$^\ast$-algebra associated to $G$ defined in \cite{MAN73}, which encodes the Weyl canonical commutation relations. For this reason, we will indicate $\mC^\ast(\mB(G))$ with the symbol $\mathcal{W}(G)$. In particular, this observation entails that the quantization prescription considered in \cite{BDHS14, BPhys} is equivalent to the one adopted here. 
\end{remark}

The preceding analysis can be applied to the scenario of interest to us. Introducing 
\begin{align}\nonumber
\mW(\Conf^k(M;\bZ)) & \doteq C^\ast(\mB(\Conf^k(M;\bZ))), & \mW(\Dyn^k(M)) & \doteq C^\ast(\mB(\Dyn^k(M))),\\
\mW(\Topf^k(M)) & \doteq C^\ast(\mB(\Topf^k(M))), & \mW(\Topt^k(M)) & \doteq C^\ast(\mB(\Topt^k(M)))
\end{align}
and recalling the symplectically orthogonal decomposition in \eqref{eqDecomposition}, as well as Proposition \ref{C*-isomorphism}, we conclude that the C$^\ast$-algebra of observables $\mW(\Conf^k(M;\bZ))$ can be factorized as a $\check\otimes$-tensor product of three contributions:

\begin{corollary}\label{isomorphism_dynamical_algebras}
	The following is an isomorphism of C$^*$-algebras:
	\begin{equation}
	\mW(\Conf^k(M;\bZ)) \simeq \mW(\Dyn^k(M)) \check\otimes \mW(\Topf^k(M)) \check\otimes \mW(\Topt^k(M)).
	\end{equation}
\end{corollary}

\begin{remark}\label{rem:dualitiesQuant}
The dualities $\zeta$, $\zeta_\Dyn$, $\zeta_\free$, $\zeta_\tor$, 
cf.\ \eqref{eqn:duality}, \eqref{eqn:dualityDyn}, \eqref{eqn:dualityFree}, \eqref{eqn:dualityTor}, 
have quantum counterparts
\begin{subequations}\label{eqn:dualitiesQuant}
\begin{align}
\mW(\zeta): \mW(\Conf^k(M;\bZ)) & \longrightarrow \mW(\Conf^{m-k}(M;\bZ)), 
& W(h,\tilde h) & \longmapsto W(\zeta(h, \tilde h)), \\
\mW(\zeta_\Dyn): \mW(\Dyn^k(M;\bZ)) & \longrightarrow \mW(\Dyn^{m-k}(M;\bZ)), 
& W(\dd A) & \longmapsto W(\zeta_\Dyn(\dd A)), \\
\mW(\zeta_\free): \mW(\Topf^k(M;\bZ)) & \longrightarrow \mW(\Topf^{m-k}(M;\bZ)), 
& W(u,\tilde u, z, \tilde z) & \longmapsto W(\zeta_\free(u,\tilde u, z, \tilde z)), \\
\mW(\zeta_\tor): \mW(\Topt^k(M;\bZ)) & \longrightarrow \mW(\Topt^{m-k}(M;\bZ)), 
& W(t, \tilde t) & \longmapsto W(\zeta_\tor(t, \tilde t)) 
\end{align}
\end{subequations}
at the C$^\ast$ algebra level defined as the unique extensions of the obvious formulas given on generators. 
The compatibility between splittings and dualities stated in Theorem \ref{thmSplitting} 
and in Remark \ref{rem:duality_m=2k} induce an analogous property 
between the factorization of Corollary \ref{isomorphism_dynamical_algebras} 
and the quantum dualities of \eqref{eqn:dualitiesQuant}. 
\end{remark}

\subsection{States for the dynamical sector}\label{Sec:states_dyn}
With Corollary \ref{isomorphism_dynamical_algebras} we have established 
a factorization of the algebra of observables 
induced by the symplectically orthogonal splitting in \eqref{eqDecomposition}. 
This allows us to define states on $\mW(\Conf^k(M;\bZ))$ 
by assigning a state on each of the $\check\otimes$-tensor factors. 
We start from the dynamical sector, cf.\ Section \ref{secDyn}, 
i.e.\ we look for a Hadamard state on $\mW(\Dyn^k(M))$. 
Our approach is motivated by the following proposition, 
where the requirement of a compact Cauchy surface is inessential and can be easily removed 
by introducing differential forms with timelike compact support, 
see e.g.\ \cite{Ben16} for analogous results in this more general case 
(we refrain from this level of generality here): 

\begin{proposition}\label{prpDynIso}
Let $M$ be an $m$-dimensional globally hyperbolic spacetime (admitting a compact Cauchy surface)
and consider the causal propagator $G: \Omega_\c^p(M) \to \Omega^p(M)$ 
for the normally hyperbolic differential operator $\Box \doteq \de \dd + \dd \de$ 
defined on $p$-forms (see \cite{BGP, B15}). Then the following is an isomorphism of vector spaces: 
\begin{align}
L: \frac{\Omega_\c^k(M)}{\Omega_{\c,\dd}^k(M) \oplus \Omega_{\c,\de}^k(M)} \longrightarrow \Dyn^k(M), 
&& [\rho] \longmapsto \dd (G \de \rho) = \ast \dd \big((-1)^{mk+1} G \ast \dd \rho\big),
\end{align}
where the subscripts $_{\dd}$ and $_{\de}$ denote 
the kernels of $\dd: \Omega_\c^k(M) \to \Omega_\c^{k+1}(M)$ 
and respectively of $\de: \Omega_\c^k(M) \to \Omega_\c^{k-1}(M)$. 
\end{proposition}

\begin{proof}
First of all, notice that $L$ is well-defined. In fact, this follows from the fact 
that $G$ is the causal propagator for $\Box$ on $p$-forms and 
that both $\dd$ and $\ast$ intertwine $\Box$ (defined on forms of suitable degrees). 
In particular, one obtains $\dd G \de = G (\Box - \de \dd) = - \de G \dd$ 
on $k$-forms with compact support. This confirms that $L$ is well-defined 
and that the equality displayed in its definition holds true. 
\vskip .5em

To confirm injectivity, let us consider $\rho \in \Omega_\c^k(M)$ such that $G \dd \de \rho = 0$. 
Then by the properties of the causal propagator, see e.g.\ \cite{BGP}, 
there exists $\alpha \in \Omega_\c^k(M)$ such that $\Box \alpha = \dd \de \rho$. 
In particular, $\dd \alpha = 0$ and $\de \rho = \de \alpha$. 
Since also $G \de \dd \rho = 0$, a similar argument allows us to find $\tilde \alpha \in \Omega_\c^k(M)$ 
such that $\de \tilde \alpha = 0$ and $\dd \rho = \dd \tilde \alpha$. 
Combining these results, one has the identity 
\begin{equation}
\Box \rho = \de \dd \alpha + \dd \de \tilde \alpha = \Box(\alpha + \tilde \alpha), 
\end{equation}
therefore $\rho = \alpha + \tilde \alpha \in \Omega_{\c,\dd}^k(M) \oplus \Omega_{\c,\de}^k(M)$. 
\vskip .5em

To show that $L$ is also surjective, consider $\dd A = \ast \dd \tilde A \in \Dyn^k(M)$. 
Without loss of generality, we can assume that $\de A = 0$ and $\de \tilde A = 0$ 
(this corresponds to fixing the Lorenz gauge, as in \cite{Ben16} for example). 
With this further assumption, $\dd A = \ast \dd \tilde A$ entails that $\Box A = 0$ and $\Box \tilde A = 0$. 
Therefore we find $\alpha \in \Omega_\c^{k-1}(M)$ and $\tilde \alpha \in \Omega_\c^{m-k-1}(M)$ 
such that $G \alpha = A$ and $G \tilde \alpha = \tilde A$. 
From $\dd A = \ast \dd \tilde A$ it follows that there exists $\rho \in \Omega_\c^k(M)$ 
such that $\dd \alpha - \ast \dd \tilde \alpha = \Box \rho$. 
Evaluating the left and the right hand side on $\dd \de$, 
one obtains $\Box \dd \alpha = \dd \de \dd \alpha = \Box \dd \de \rho$, 
hence $\dd \alpha = \dd \de \rho$. This allows us to conclude that $\dd A = \dd G \de \rho$. 
Since $L$ is clearly linear, we conclude that $L$ is an isomorphism of vector spaces as claimed. 
\end{proof}

The isomorphism in Proposition \ref{prpDynIso} can be promoted to one of symplectic vector spaces. 
In fact, we can equip the vector space 
\begin{equation}
\Omega_\c^k(M)_\Dyn \doteq \frac{\Omega_\c^k(M)}{\Omega_{\c,\dd}^k(M) \oplus \Omega_{\c,\de}^k(M)}
\end{equation}
with a symplectic structure as follows: 
\begin{align}
\tau_\Dyn: \Omega_\c^k(M)_\Dyn \times \Omega_\c^k(M)_\Dyn \longrightarrow \bR,
&& ([\rho],[\rho^\prime]) \longmapsto \int_M \rho \wedge \ast G \dd \de \rho^\prime. 
\end{align}
Notice that there are other equivalent formulas defining $\tau_\Dyn$: 
\begin{equation}
-\int_M \dd \rho \wedge \ast G \dd \rho^\prime = \int_M \rho \wedge \ast G \dd \de \rho^\prime 
= \int_M \de \rho \wedge \ast G \de \rho^\prime. 
\end{equation}
These identities, which follow from the fact that $\dd$ and its formal adjoint $\de$ 
intertwine the causal propagators, show that $\tau_\Dyn$ is well-defined. 
To confirm that $\tau_\Dyn$ is antisymmetric recall that 
$\Box$ is formally self-adjoint, hence the causal propagator $G$ is formally anti-selfadjoint. 
Being also non-degenerate (to prove it one argues as for injectivity in Proposition \ref{prpDynIso}), 
$\tau_\Dyn$ is indeed a symplectic form. 
With a quite standard, although lengthy, computation, one checks that 
the isomorphism $L$ is compatible with the symplectic structures $\tau_\Dyn$ and $\sigma_\Dyn$, 
respectively defined on the source and on the target. 
This calculation is based on Stokes theorem and on the properties of the retarded/advanced Green operators 
$G^\pm: \Omega_\c^p \to \Omega^p(M)$ for $\Box: \Omega^p(M) \to \Omega^p(M)$: 
\begin{align}\label{eqn:tausigma}
\tau_\Dyn([\rho],[\rho^\prime]) & = \int_{J^+_M(\Sigma)} \Box G^- \rho \wedge \ast \dd G \de \rho^\prime 
+ \int_{J^-_M(\Sigma)} \Box G^+ \rho \wedge \ast \dd G \de \rho^\prime \nonumber \\
& = \int_\Sigma \de G \rho \wedge \ast \dd G \de \rho^\prime 
- \int_\Sigma \dd G \de \rho^\prime \wedge \ast \dd G \rho \nonumber \\
& = \sigma_\Dyn(\dd(G \de \rho), \dd(G \de \rho^\prime)),
\end{align}
where $\Sigma$ is a Cauchy surface of $M$ while $J^\pm_M(\Sigma)$ denotes its causal future/past. 
For later reference, let us observe how $\zeta_\Dyn$ of \eqref{eqn:dualityDyn} looks like from this point of view:
\begin{align}\label{eqn:dualityDyn2}
\zeta_\Dyn: \Omega^k_\c(M)_\Dyn \longrightarrow \Omega^{m-k}_\c(M)_\Dyn, 
&& [\rho] \longmapsto [(-1)^{k(m-k)} \ast \rho]. 
\end{align}
This alternative, yet equivalent, perspective on the symplectic vector space $\Dyn^k(M)$ 
suggests us how to introduce a two-point function that will be later used to define 
a Hadamard state on $\mW(\Dyn^k(M))$. In fact, due to \cite{SV01}, 
one always obtains a Hadamard two-point function $\mathfrak{W}_k \in \Omega_\c^{2k}(M \times M)^\prime$ 
associated to $\Box: \Omega^k(M) \to \Omega^k(M)$. For example, when dealing 
with ultra-static spacetimes, one way to achieve this result is to adopt 
the so-called \textit{positive frequencies} prescription, which leads to the ground state (see e.g.\ \cite{WAL94}). 
Then, mimicking the formula for the symplectic form $\tau_\Dyn$, 
one is induced to regard $\mathfrak{W}_k \circ (\id \otimes \dd \de)$ as a natural candidate 
for the two-point function of the quantum field theory corresponding to $\Dyn^k(M)$. 
\vskip 1em

For the sake of concreteness, let us focus on the case of an ultra-static globally hyperbolic spacetime $M$ 
admitting a compact Cauchy surface $\Sigma$. This means that we can present $M$ as 
\begin{align}
M \simeq \bR \times \Sigma, && g = -\dd t \otimes \dd t + h, 
\end{align}
where $h$ is a Riemannian metric on $\Sigma$ (constant in $t \in \bR$). 
This allows us to decompose differential forms on $M$ in terms of sections 
of the pullbacks along the projection $\pi_2: M \to \Sigma$ of the bundles $\bigwedge^p T^\ast \Sigma$ 
of skew-symmetric $p$-cotensors over $\Sigma$. Specifically, one has: 
\begin{equation}
\Omega^k(M) = \Gamma\left(M,\pi_2^\ast \bigwedge^k T^\ast \Sigma\right) 
\oplus dt \wedge \Gamma\left(M,\pi_2^\ast \bigwedge^{k-1} T^\ast \Sigma\right). 
\end{equation}
With respect to this decomposition, $\Box$ takes the form 
\begin{equation}
\Box(\omega_\Sigma + \dd t \wedge \omega_t) = (\partial_t^2 \omega_\Sigma + \triangle \omega_\Sigma) 
+ \dd t \wedge (\partial_t^2 \omega_t + \triangle \omega_t).
\end{equation}
This allows us to use the spectral calculus associated to 
the Hodge-de Rham Laplacian $\triangle = \de_\Sigma \dd_\Sigma + \dd_\Sigma \de_\Sigma$ on $\Sigma$ 
(note that differential, codifferential and Hodge dual are indicated with a subscript $_\Sigma$ 
whenever they refer to the geometry of the Cauchy surface $\Sigma$, instead of that of the whole spacetime $M$). 
In particular, for $p$-forms on $\Sigma$ we have the Hodge decomposition 
into harmonic, exact and coexact contributions: 
\begin{equation}
\Omega^p(\Sigma) = \mH^p(\Sigma) \oplus \dd_\Sigma \Omega^{p-1}(\Sigma) \oplus \de_\Sigma \Omega^{p+1}(\Sigma). 
\end{equation}
We denote the projections on the harmonic part and the projection on its orthogonal complement by: 
\begin{align}\label{projections}
\pi_\mH^p: \Omega^p(\Sigma) \longrightarrow \mH^p(\Sigma), &&
\pi_\perp^p: \Omega^p(\Sigma) \longrightarrow \dd_\Sigma \Omega^{p-1}(\Sigma) 
\oplus \de_\Sigma \Omega^{p+1}(\Sigma).
\end{align}
With these preparations, we can write down a quite explicit formula for the causal propagator $G$ 
associated to $\Box$ acting on $\Omega^k(M)$. Regarding $\rho_\Sigma$ and $\rho_t$ 
as smoothly $\bR$-parametrized differential forms on $\Sigma$, 
i.e.\ $t \in \bR \mapsto \rho_\Sigma(t, \cdot) \in \Omega^k(\Sigma)$ 
and $t \in \bR \mapsto \rho_t(t, \cdot) \in \Omega^{k-1}(\Sigma)$, we obtain 
\begin{subequations}
\begin{align}
G: \Omega_\c^k(M) & \longrightarrow \Omega^k(M), \nonumber \\
\rho_\Sigma + \dd t \wedge \rho_t & \longmapsto 
G_\mH\, \pi_\mH^k\, \rho_\Sigma + G_\perp\, \pi_\perp^k\, \rho_\Sigma 
+ \dd t \wedge (G_\mH\, \pi_\mH^{k-1}\, \rho_t + G_\perp\, \pi_\perp^{k-1}\, \rho_t) ,
\end{align}
where 
\begin{align}
(G_\mH\, \alpha_\mH)(t,\cdot) & = \int_\bR (t - t^\prime)\, \alpha_\mH(t^\prime, \cdot) \dd t^\prime, \\
(G_\perp\, \alpha_\perp)(t,\cdot) & = \int_\bR 
\triangle^{-\frac{1}{2}} \sin(\triangle^\frac{1}{2} (t - t^\prime))\, \alpha_\perp(t^\prime, \cdot) \dd t^\prime, 
\end{align}
\end{subequations}
for $\alpha_\mH, \alpha_\perp \in \Gamma_\c(M,\pi_2^\ast \bigwedge^p T^\ast \Sigma)$ such that, 
for each $t \in \bR$, $\alpha_\mH(t, \cdot) \in \mH^p(\Sigma)$ and 
$\alpha_\perp(t, \cdot) \in \dd_\Sigma \Omega^{p-1}(\Sigma) \oplus \de_\Sigma \Omega^{p+1}(\Sigma)$. 

Following an approach inspired by \cite{FP03}, we introduce a bidistribution 
$\mathfrak{W}_k \in \Omega_\c^{2k}(M \times M)^\prime$ where only the part orthogonal to the harmonic one contributes: 
\begin{subequations}\label{state-part1}
\begin{align}\nonumber
\mathfrak{W}_k: \Omega_\c^k(M) \otimes \Omega_\c^k(M) & \longrightarrow \bC \\
(\rho_\Sigma + \dd t \wedge \rho_t) \otimes (\rho^\prime_\Sigma + \dd t \wedge \rho^\prime_t) 
& \longmapsto \mathfrak{W}_\perp(\pi_\perp^k\, \rho_\Sigma \otimes \pi_\perp^k\, \rho^\prime_\Sigma) 
- \mathfrak{W}_\perp(\pi_\perp^{k-1}\, \rho_t \otimes \pi_\perp^{k-1}\, \rho^\prime_t),
\end{align}
where 
\begin{equation}
\mathfrak{W}_\perp(\alpha_\perp \otimes \alpha^\prime_\perp) = \int_\bR \int_\bR \Big\langle \alpha_\perp(t, \cdot),\, 
\tfrac{1}{2} \triangle^{-\frac{1}{2}} \exp(-i\, \triangle^\frac{1}{2}\, (t-t^\prime))\, 
\alpha^\prime_\perp(t^\prime, \cdot) \Big\rangle\, \dd t\, \dd t^\prime,\label{state-part2}
\end{equation}
\end{subequations}
for $\alpha_\perp, \alpha^\prime_\perp \in \Gamma_\c(M,\pi_2^\ast \bigwedge^p T^\ast \Sigma)$ ($p=k-1,k$) 
such that, for each $t \in \bR$, $\alpha_\perp(t, \cdot), \alpha_\perp^\prime(t, \cdot) 
\in \dd_\Sigma \Omega^{p-1}(\Sigma) \oplus \de_\Sigma \Omega^{p+1}(\Sigma)$, 
and where $\langle \cdot, \cdot \rangle$ denotes the $L^2$-scalar product on $\Omega^p(\Sigma)$. 
A straightforward computation allows us to confirm that $\mathfrak{W}_k$ is a bisolution of $\Box$. 
In fact, for $\alpha_\perp, \alpha^\prime_\perp$, as above we have 
\begin{equation}
\mathfrak{W}_\perp((\partial_t^2 \alpha_\perp + \triangle \alpha_\perp) \otimes \alpha^\prime_\perp) = 
\mathfrak{W}_\perp(\alpha_\perp \otimes (\partial_t^2 \alpha^\prime_\perp + \triangle \alpha^\prime_\perp)) = 0. 
\end{equation} 
The argument illustrated in \cite[Appendix B]{FP03} allows us to conclude that 
$\mathfrak{W}_k$ fulfils the microlocal spectrum condition (recall that $\mathfrak{W}_k$ is a $\Box$-bisolution 
whose antisymmetric part differs from $-i\, G$ only for the harmonic contribution, which is smooth). 
Using $\mathfrak{W}_k$ we introduce 
\begin{equation}
\omega_2 \doteq (\id \otimes \dd \de)\, \mathfrak{W}_k 
= \mathfrak{W}_k \circ (\id \otimes \dd \de) \in \Omega_\c^{2(m-k)}(M \times M)^\prime. 
\end{equation}
Recalling that for each $\omega \in \Omega^k(M)$ one has 
\begin{subequations}
\begin{align}
\dd \de (\omega_\Sigma + \dd t \wedge \omega_t) 
= \dd_\Sigma \de_\Sigma \omega_\Sigma + \dd_\Sigma \partial_t \omega_t 
+ \dd t \wedge (\partial_t^2 \omega_t + \dd_\Sigma \de_\Sigma \omega_t 
+ \de_\Sigma \partial_t \omega_\Sigma), \\
\de \dd (\omega_\Sigma + \dd t \wedge \omega_t) = 
\partial_t^2 \omega_\Sigma + \de_\Sigma \dd_\Sigma \omega_\Sigma - \dd_\Sigma \partial_t \omega_t 
+ \dd t \wedge (\de_\Sigma \dd_\Sigma \omega_t - \de_\Sigma \partial_t \omega_\Sigma), 
\end{align}
\end{subequations}
one can confirm that 
\begin{equation}
\omega_2 = (\dd \de \otimes \id)\, \mathfrak{W}_k. 
\end{equation}
Furthermore, since $\mathfrak{W}_k$ is a bisolution of $\Box$, it follows that 
\begin{equation}
\omega_2 = - (\id \otimes \de \dd)\, \mathfrak{W}_k = - (\de \dd \otimes \id)\, \mathfrak{W}_k. 
\end{equation}
Similarly, recalling that for each $\omega \in \Omega^p(M)$ one has 
\begin{subequations}\label{eqn:de}
\begin{align}
\dd (\omega_\Sigma + \dd t \wedge \omega_t) 
& = \dd_\Sigma \omega_\Sigma + \dd t \wedge (\partial_t \omega_\Sigma - \dd_\Sigma \omega_t), \\
\de (\omega_\Sigma + \dd t \wedge \omega_t) 
& = \de_\Sigma \omega_\Sigma + \partial_t \omega_t - \dd t \wedge \de_\Sigma \omega_t, 
\end{align}
\end{subequations}
one shows also that 
\begin{equation}\label{eqn:twopoint-dede}
\omega_2 = (\de \otimes \de)\, \mathfrak{W}_{k-1} = - (\dd \otimes \dd)\, \mathfrak{W}_{k+1}. 
\end{equation}
These observations entail that $\omega_2$ vanishes both on closed and on coclosed forms, thus yielding 
\begin{align}
\omega_2: \Omega_\c^k(M)_\Dyn \otimes \Omega_\c^k(M)_\Dyn \longrightarrow \bC, 
&& [\rho] \otimes [\rho^\prime] \longmapsto \omega_2(\rho \otimes \rho^\prime). 
\end{align}
Notice that the antisymmetric part of $\omega_2$ agrees with $\tau_\Dyn$: 
\begin{equation}
\omega_2([\rho] \otimes [\rho^\prime] - [\rho^\prime] \otimes [\rho]) = -i\, \tau_\Dyn([\rho],[\rho^\prime]) 
\end{equation}
for all $[\rho],[\rho^\prime] \in \Omega^k(M)_\Dyn$. 
Eq.\ \eqref{eqn:twopoint-dede} is also crucial to confirm that $\omega_2$ inherits 
the microlocal spectrum condition from $\mathfrak{W}_{k+1}$. In fact, we are going to show that 
\begin{equation}\label{eqn:WF}
\WF(\omega_2) = \WF(\mathfrak{W}_{k+1}). 
\end{equation}
First of all, notice that the principal symbol of $\dd \otimes \dd$ is the 
homomorphism of vector bundles over $M \times M$: 
\begin{align}
\sigma_{\dd \otimes \dd}: T^\ast (M \times M) \otimes 
\Big( \bigwedge^k T^\ast M \boxtimes \bigwedge^k T^\ast M \Big) 
& \longrightarrow \bigwedge^{k+1} T^\ast M \boxtimes \bigwedge^{k+1} T^\ast M) \nonumber \\
(k,k^\prime) \otimes (\omega \otimes \omega^\prime) 
& \longmapsto (k \wedge \omega) \otimes (k^\prime \wedge \omega^\prime). 
\end{align}
In particular, it follows that $(k,k^\prime) \in  T^\ast (M \times M)$ belongs to 
the characteristic set $\Char(\dd \otimes \dd)$ 
if and only if precisely one between $k$ and $k^\prime$ vanishes. 
Then the microlocal spectrum condition for $\mathfrak{W}_{k+1}$ entails that 
$\WF(\mathfrak{W}_{k+1}) \cap \Char(\dd \otimes \dd) = \emptyset$. 
Taking into account also \cite[Ch.\ 8]{Hor}, we have the chain of inclusions 
$\WF(\omega_2) \subseteq \WF(\mathfrak{W}_{k+1}) \subseteq \WF(\omega_2) \cup \Char(\dd \otimes \dd)$. 
Therefore \eqref{eqn:WF} follows, showing that $\omega_2$ 
inherits the microlocal spectrum condition from $\mathfrak{W}_{k+1}$. 
\vskip 1em

To summarize, we constructed a bidistribution $\omega_2$ that fulfils the microlocal spectrum condition, 
that descends to the quotient in $\Omega_\c^k(M)_\Dyn$ and that is compatible 
with the canonical commutation relations encoded in $\mW(\Dyn^k(M))$. 
A straightforward computation conducted expanding \eqref{eqn:twopoint-dede} 
allows us to confirm that $\omega_2$ is also non-negative, cf.\ \cite{FP03} for a similar argument: 
\begin{equation}
\omega_2([\rho] \otimes [\rho]) \geq 0
\end{equation}
for all $[\rho] \in \Omega_\c^k(M)_\Dyn$. In fact, introducing also the projections 
\begin{align}
\pi_{\dd}^p: \Omega^p(\Sigma) \longrightarrow \dd_\Sigma \Omega^{p-1}(\Sigma), 
&& \pi_{\de}^p: \Omega^p(\Sigma) \longrightarrow \de_\Sigma \Omega^{p+1}(\Sigma),
\end{align}
that decompose $\pi^p_\perp$ as $\pi^p_\perp = (\pi_{\dd}^p, \pi_{\de}^p)$ 
and recalling \eqref{eqn:de}, for all $\rho, \rho^\prime \in \Omega_\c^k(M)$ one obtains 
\begin{align}\label{2-pt_function}
\omega_2(\rho \otimes \rho^\prime) & = \mathfrak{W}_{k-1}(\de \rho \otimes \de \rho^\prime) \nonumber \\
& = \mathfrak{W}_\perp \Big( \pi_{\de}^{k-1}(\de_\Sigma \rho_\Sigma + \partial_t \rho_t) 
\otimes \pi_{\de}^{k-1}(\de_\Sigma \rho^\prime_\Sigma + \partial_t \rho^\prime_t) \Big) \nonumber \\
& \phantom{=} + \mathfrak{W}_\perp \Big( \pi_{\dd}^{k-1} (\partial_t \rho_t) \otimes \pi_{\dd}^{k-1} (\partial_t \rho^\prime_t) \Big) 
- \mathfrak{W}_\perp \Big( \pi_{\de}^{k-2}(\de_\Sigma \rho_t) \otimes \pi_{\de}^{k-2}(\de_\Sigma \rho^\prime_t) \Big) \nonumber \\
& = \mathfrak{W}_\perp \Big( \pi_{\de}^{k-1}(\de_\Sigma \rho_\Sigma + \partial_t \rho_t) 
\otimes \pi_{\de}^{k-1}(\de_\Sigma \rho^\prime_\Sigma + \partial_t \rho^\prime_t) \Big) \geq 0, 
\end{align}
where the two contributions appearing in the third line cancel out due to 
\begin{align}
\mathfrak{W}_\perp \Big( \pi_{\de}^{k-2}(\de_\Sigma \rho_t) \otimes \pi_{\de}^{k-2}(\de_\Sigma \rho^\prime_t) \Big) 
\nonumber & = \int_\bR \int_\bR \langle \pi_{\dd}^{k-1} \rho_t, \tfrac{1}{2} \triangle^\frac{1}{2}\, 
\exp(-i\, \triangle^\frac{1}{2} (t - t^\prime))\, \pi_{\dd}^{k-1} \rho^\prime_t \rangle\, \dd t \, \dd t^\prime \\
& = \mathfrak{W}_\perp \Big( \pi_{\dd}^{k-1} (\partial_t \rho_t) \otimes \pi_{\dd}^{k-1} (\partial_t \rho^\prime_t) \Big). 
\end{align}
We are now in a position to define the desired state on $\mW(\Dyn^k(M))$: 
\begin{proposition}\label{prp:stateDyn}
Let $M$ be an $m$-dimensional globally hyperbolic spacetime with compact Cauchy surface 
and recall Proposition \ref{prpDynIso}. Then 
\begin{align}
\omega_\Dyn: \mW(\Dyn^k(M)) \longrightarrow \bC, 
&& W(L [\rho]) \longmapsto \exp \big( -2\pi\, \omega_2([\rho] \otimes [\rho]) \big). 
\end{align}
is a state on the C$^\ast$-algebra $\mW(\Dyn^k(M))$ that fulfils the microlocal spectrum condition. 
Furthermore, for the state on $\mW(\Dyn^k(M))$ and its analogue on $\mW(\Dyn^{m-k}(M))$, one has 
\begin{equation}
\omega_\Dyn \circ \mW(\zeta_\Dyn) = \omega_\Dyn, 
\end{equation}
where $\mW(\zeta_\Dyn): \mW(\Dyn^k(M)) \to \mW(\Dyn^{m-k}(M))$ 
is the duality introduced in Remark \ref{rem:dualitiesQuant}. 
\end{proposition}

\begin{proof}
We have shown that $\omega_2$ is a bidistribution 
1) that fulfils the microlocal spectrum condition, 
2) that descends to the quotient $\Omega^k(M)_\Dyn \simeq \Dyn^k(M)$, 
3) whose antisymmetric part coincides with $-i\, \tau_\Dyn$ (which is equivalent to $-i\, \sigma_\Dyn$ 
under the isomorphism of Proposition \ref{prpDynIso}, cf.\ \eqref{eqn:tausigma}), 
4) that is non-negative. 
Therefore $\omega_\Dyn$ will be a ``Hadamard state'' for the C$^\ast$-algebra $\mW(\Dyn^k(M))$ 
as soon as we confirm that it is sufficient to specify it 
on the generators of the $\ast$-algebra $\mA(\Dyn^k(M))$, cf.\ Section \ref{sec:quantization}. 
Note that it is positive and normalized on $\mA(\Dyn^k(M))$. 
Furthermore, it is immediate to check continuity with respect to $\|\cdot\|_1$ 
because the exponential factor is bounded from above by 1. 
In particular, $\omega_\Dyn$ can be extended to the $\|\cdot\|_1$-completion of $\mA(\Dyn^k(M))$, 
i.e.\ the Banach $\ast$-algebra $\mB(\Dyn^k(M))$. 
By a standard property of the enveloping C$^\ast$-algebra \cite[Prop.\ 2.7.4]{Dixmier}, 
the representations of $\mW(\Dyn^k(M)) = \mC^\ast(\mB(\Dyn^k(M)))$ are in bijective correspondence 
with those of $\mB(\Dyn^k(M))$. Therefore it is sufficient to specify $\omega_\Dyn$ on $\mA(\Dyn^k(M))$ 
(as we did) in order to obtain a unique canonical extension to $\mW(\Dyn^k(M))$. 
\vskip .5em

To confirm that our prescription for the construction of $\omega_\Dyn$ is compatible 
with $\mW(\zeta_\Dyn): \mW(\Dyn^k(M)) \to \mW(\Dyn^{m-k}(M))$, 
let us observe that, on account of a similar property of the $L^2$-scalar product on $\Omega^p(\Sigma)$, 
for all $\alpha_\perp \in \Gamma_\c(M, \pi_2^\ast \bigwedge^p T^\ast \Sigma)$ and 
$\beta_\perp \in \Gamma_\c(M, \pi_2^\ast \bigwedge^{m-p-1} T^\ast \Sigma)$ such that, for all $t \in \bR$, 
$\alpha_\perp(t, \cdot) \in \dd_\Sigma \Omega^{p-1}(\Sigma) \oplus \de_\Sigma \Omega^{p+1}(\Sigma)$ and 
$\beta_\perp(t, \cdot) \in \dd_\Sigma \Omega^{m-p-2}(\Sigma) \oplus \de_\Sigma \Omega^{m-p}(\Sigma)$, 
one has 
\begin{equation}\label{eqn:Wperpast}
\mathfrak{W}_\perp(\alpha_\perp \otimes \ast_\Sigma \beta_\perp) 
= \mathfrak{W}_\perp((-1)^{mp} \ast_\Sigma \alpha_\perp \otimes \beta_\perp), 
\end{equation}
As a direct consequence, for all $[\rho], [\rho^\prime] \in \Omega^k_\c(M)_\Dyn$ one finds 
\begin{equation}
\omega_2(\zeta_\Dyn [\rho] \otimes \zeta_\Dyn [\rho^\prime]) 
= \mathfrak{W}_{m-k-1}(\de \ast \rho \otimes \de \ast \rho^\prime) 
= - \mathfrak{W}_{k+1}( \dd \rho \otimes \dd \rho^\prime) = \omega_2([\rho] \otimes [\rho^\prime]). 
\end{equation}
Notice that we used \eqref{eqn:dualityDyn2} and \eqref{eqn:twopoint-dede} for the first step, 
\eqref{eqn:Wperpast} for the second one and again \eqref{eqn:twopoint-dede} to conclude. 
In particular, this entails the desired claimed relation 
between the state on $\mW(\Dyn^k(M))$ and the state on $\mW(\Dyn^{m-k}(M))$. 
\end{proof}

\begin{remark}
To conclude this section, we observe that $\omega_\Dyn$ has been constructed so to be a {\em ground state}, as per \cite[App. A]{SV00}.
\end{remark}

\subsection{\label{sec:TopfState}States for the torsion-free topological sector} \label{sec:states_torsion_free}
In this section we exhibit a state on $\mW(\Topf^k(M))$ commenting, in particular, on its significance. 
Recalling Section \ref{sec:quantization}, we have that $\mW(\Topf^k(M))$ 
is the enveloping C$^\ast$-algebra associated to the Banach $\ast$-algebra $\mB(\Topf^k(M))$ 
obtained as the $\|\cdot\|_1$-completion of the $\ast$-algebra $\mA(\Topf^k(M))$. 
Recalling also Section \ref{secTopf}, we denote the generators of $\mA(\Topf^k(M))$ 
by $W(u,\tilde u, z, \tilde z)$ for $(u,\tilde u, z, \tilde z) \in \Topf^k(M)$. 

\begin{proposition}\label{prp:statefree}
	Let $\omega_\free:\mW(\Topf^k(M))\to\bC$ be the linear functional specified by
	\begin{equation}\label{eqStateTopf}
	\omega_\free(W(u,\tilde u, z, \tilde z))=\begin{cases}
	1 & \textrm{if } z=0, \tilde{z}=0,\\
	0 & \textrm{otherwise}.
	\end{cases}
	\end{equation}
	Then $\omega_\free$ is a state on the C$^\ast$-algebra $\mW(\Topf^k(M))$. Furthermore the state on $\mW(\Topf^k(M))$ and its analogue on $\mW(\Topf^{m-k}(M))$ are compatible with the duality $\mW(\zeta_\free): \mW(\Topf^k(M)) \to \mW(\Topf^{m-k}(M))$ introduced in Remark \ref{rem:dualitiesQuant}, i.e.,  
	\begin{equation}
	\omega_\free \circ \mW(\zeta_\free) = \omega_\free.
	\end{equation}
\end{proposition}

\begin{proof}
	The functional is normalized since the unit in $\mathcal{W}(\Topf^k(M))$ is the element $W(0,0,0,0)$ and, by definition, $\omega_\free(W(0,0,0,0))=1$. To prove positivity, let $I$ be an index set of finite cardinality and let $a=\sum_{i\in I}\alpha_i\, W(u_i,\tilde{u}_i,z_i,\tilde{z}_i)$, where $\alpha_i \in \bC$ and $(u_i,\tilde{u}_i,z_i,\tilde{z}_i) \in \Topf^k(M)$ for all $i \in I$. Without loss of generality, we assume $(u_i,\tilde{u}_i,z_i,\tilde{z}_i)\neq(u_j,\tilde{u}_j,z_j,\tilde{z}_j)$ for all $i,j\in I$ such that $i\neq j$. Set $i\sim j$ if and only if $z_i=z_j$ and $\tilde{z}_i=\tilde{z}_j$. Clearly, $\sim$ is an equivalence relation. Let $\tilde{I}=I/\sim$ and let us indicate with $\tilde{i}$ the equivalence class of $i \in I$. Using \eqref{eqTopfSympl} and \eqref{eqStateTopf}, we obtain
	\begin{equation}
	\omega_\free(a^*\, a) = \sum_{\tilde{i} \in \tilde{I}} 	\Big| \sum_{i \in \tilde{i}} \alpha_i\, 
	\exp\Big(2\pi i\, \sigma\big( (u_i, \tilde u_i, 0, 0), (0, 0, z_i, \tilde z_i) \big) \Big) \Big|^2 \geq 0,
	\end{equation}
    which guarantees the positivity of $\omega_\free$. Furthermore, $\omega_\free$ is clearly continuous with respect to the norm $\|\cdot\|_1$, hence it induces a unique state on the Banach $\ast$-algebra $\mB(\Topf^k(M))$. By \cite[Prop.\ 2.7.4]{Dixmier} this provides a unique representation, hence a state, also on the enveloping C$^\ast$-algebra $\mW(\Topf^k(M)) = \mC^\ast(\mB(\Topf^k(M)))$. 
\vskip .5em

To confirm that our prescription is compatible with duality note that the last two components of $\zeta_\free(u, \tilde u, z, \tilde z) \in \Topf^{m-k}(M)$ vanish if and only if the last two components of $(u, \tilde u, z, \tilde z) \in \Topf^k(M)$ vanish. Therefore $\omega_\free(\mW(\zeta_\free)\, W(u, \tilde u, z, \tilde z)) = \omega_\free(W(u, \tilde u, z, \tilde z))$, leading to the conclusion. 
\end{proof}

Observe that the state is not faithful: by direct inspection of \eqref{eqStateTopf} one finds $0 \neq a \in \mathcal{W}(\Topf^k(M))$ such that $\omega_\free(a^\ast\, a) = 0$. For example, such an $a$ is given by 
\begin{equation}\label{Gelfand_quotient}
a = W(0,0,z,\tilde{z}) 
- \exp \Big( 2\pi i\, \sigma_\free \big( (0, 0, z, \tilde z), (u, \tilde u, z, \tilde z) \big) \Big)\, W(u,\tilde u,z,\tilde z).
\end{equation}

\begin{remark}\label{rem:faithstatefree}
	A faithful alternative to $\omega_\free$ is the state $\widetilde{\omega}_\free: \mW(\Topf^k(M)) \to \bC$ defined by 
	\begin{equation}\label{eqStateTopf_faithful}
	\widetilde{\omega}_\free(W(u,\tilde{u},z,\tilde{z}))=
	\begin{cases}
	1 & \textrm{if } u=0, \tilde{u}=0, z=0, \tilde{z}=0,\\
	0 & \textrm{otherwise}.
	\end{cases}
	\end{equation}
	Although the GNS representation induced by $\widetilde \omega_\free$ is faithful, $\omega_\free$ leads to a more appealing interpretation, which is why we regard it as our prime example.
\end{remark}

In order to explain why we regard \eqref{eqStateTopf} as our prime example, we construct the associated GNS representation. The Gelfand ideal $\mI^k_\free \subseteq \mW(\Topf^k(M))$ of $\omega_\free$ is precisely generated by elements of $\mW(\Topf^k(M))$ of the form \eqref{Gelfand_quotient}. Hence the GNS Hilbert space is the completion 
\begin{subequations}
\begin{equation}
\scrH^k_\free \doteq \overline{\scrD^k_\free}
\end{equation}
of the pre-Hilbert space 
\begin{equation}
\scrD^k_\free \doteq \mW(\Topf^k(M)) / \mI^k_\free 
= \textrm{span}_\bC \{ |z,\tilde{z}\rangle : (z,\tilde z) \in \H^{k,m-k}(M;\bZ) \}
\end{equation}
equipped with the scalar product 
\begin{align}
\langle 	\,\cdot	\, | \, \cdot \, \rangle: \scrD^k_\free \times \scrD^k_\free \longrightarrow \bC, 
&& \langle z^\prime,\tilde{z}^\prime | z,\tilde{z} \rangle 
\doteq \omega_\free \big(W(0,0,z^\prime,\tilde{z}^\prime)^\ast\, W(0,0,z,\tilde{z}) \big)
\end{align}
induced by $\omega_\free$, where for notational convenience we set 
\begin{equation}
| z,\tilde{z} \rangle \doteq [W(0,0,z,\tilde{z})] \in \mW(\Topf^k(M)) / \mI^k_\free. 
\end{equation}
\end{subequations}
The GNS representation associated to $\omega_\free$ is defined by 
\begin{subequations}
\begin{align}
\pi^k_\free: \mW(\Topf^k(M)) \longrightarrow \mathcal{BL}(\scrH^k_\free), 
&& W(u,\tilde{u},z,\tilde{z})\, \longmapsto \pi_\free^k(W(u,\tilde{u},z,\tilde{z})),
\end{align}
where $\pi_\free^k(W(u,\tilde{u},z,\tilde{z}))$ acts on $\scrH^k_\free$ according to 
\begin{align}
\pi_\free^k(W(u,\tilde{u}, z,\tilde{z})): \scrH^k_\free & \longrightarrow \scrH^k_\free, \\
| z^\prime,\tilde{z}^\prime \rangle & \longmapsto 
\exp \Big( 2\pi i\, \sigma_\free \big( (u, \tilde u, 0, 0), (0, 0, z + 2 z^\prime, \tilde z + 2 \tilde z^\prime) \big) \Big)\, |z+z^\prime,\tilde{z}+\tilde{z}^\prime \rangle. \nonumber
\end{align}
\end{subequations}
As a by-product, the cyclic vector of the GNS representation is $|0,0\rangle$. 
Furthermore, one observes that generators of the form $W(u, \tilde u , 0, 0)$ act on $\scrH^k_\free$ 
by multiplication with a phase that depends linearly on $u$ and $\tilde u$, 
while those of the form $W(0, 0, z, \tilde z)$ act on $\scrH^k_\free$ by shift: 
\begin{align}\label{eqn:phase_shift}
\Phi^k(u, \tilde u) \doteq \pi_\free^k(W(u, \tilde u, 0, 0)), 
&& \Sigma^k(z, \tilde z) \doteq \pi_\free^k(W(0, 0, z, \tilde z)). 
\end{align}
In particular, it holds that 
\begin{align}
\Phi^k: \H^{k-1,m-k-1}(M;\bR) / \Hf^{k-1,m-k-1}(M;\bZ) \longmapsto \mathcal{BL}(\scrH^k_\free), 
&& (u, \tilde u) \longmapsto \Phi^k(u, \tilde u) 
\end{align}
is a strongly continuous family of unitary operators linearly 
parametrized by the quotient of $\H^{k-1,m-k-1}(M;\bR)$ by $\Hf^{k-1,m-k-1}(M;\bZ)$. 
In particular, for each $(r, \tilde r) \in \H^{k-1,m-k-1}(M;\bR)$, 
Stone's theorem provides an unbounded densely defined self-adjoint operator 
\begin{align}\label{eqn:fluxobs}
P^k(r, \tilde r): \scrD^k_\free \to \scrH^k_\free, && |z, \tilde z\rangle \longmapsto 
{2\, \widetilde \sigma_\free \big( (r, \tilde r, 0, 0), (0, 0, z, \tilde z) \big)}\, |z, \tilde z\rangle
\end{align}
that generates $t \in \bR \mapsto \Phi^k\left(t(r, \tilde r)\right) = \exp\big(2 \pi i t\, P^k(r, \tilde r)\big)$. 
Here $\widetilde \sigma_\free$ is the lift of $\sigma_\free$ defined by 
\begin{align}
\widetilde \sigma_\free: \big( \H^{k-1,m-k-1}(M;\bR) \oplus \Hf^{k,m-k}(M;\bZ) \big)^{\times 2} & \longrightarrow \bR, \\
\big( (r, \tilde r, z, \tilde z), (r^\prime, \tilde r^\prime, z^\prime, \tilde z^\prime) \big) 
\longmapsto & (\iota_\Sigma^\ast \tilde r \smile \iota_\Sigma^\ast z^\prime)[\Sigma] 
- (-1)^{k(m-k)} (\iota_\Sigma^\ast r \smile \iota_\Sigma^\ast \tilde z^\prime)[\Sigma] \nonumber \\
& - (\iota_\Sigma^\ast \tilde r^\prime \smile \iota_\Sigma^\ast z)[\Sigma] 
+ (-1)^{k(m-k)} (\iota_\Sigma^\ast r^\prime \smile \iota_\Sigma^\ast \tilde z)[\Sigma] \nonumber
\end{align}
in terms of the cohomological pairing on $\Sigma$ introduced in \eqref{eqn:realcohopairing}. 
Notice that these operators precisely detect the values of $z$ and $\tilde z$, 
which correspond to magnetic and electric fluxes \cite{BPhys, FMSa, FMSb}. 
As such, we regard the operators $P^k(r, \tilde r)$ as \textit{flux observables} 
(for $r=0$ only the magnetic flux is tested, conversely for $\tilde r = 0$ only the electric flux). 
The shift operators instead are precisely those modifying such fluxes 
(adding $z$ to the magnetic flux and $\tilde z$ to the electric one). 
Because of this appealing interpretation, that resembles the quantum mechanical description of 
a system formed by point particles freely moving on the circle with momenta $(z, \tilde z)$, 
we regard $\omega_\free$ as our prime example of state for the torsion-free topological sector. 

\begin{remark}\label{rem:statefree}
We already observed that the duality $\mW(\zeta_\free): \mW(\Topf^k(M)) \to \mW(\Topf^{m-k}(M))$ 
preserves the states $\omega_\free$ defined on the source and the target, cf.\ Proposition \ref{prp:statefree}. 
As a consequence, we obtain the isomorphism 
\begin{align}
U_\free^k: \scrH^k_\free \longrightarrow \scrH^{m-k}_\free, 
&& | z, \tilde z \rangle \longmapsto | \tilde z, (-1)^{k(m-k)+1} z\rangle 
\end{align}
between the GNS Hilbert spaces. $U_\free^k$ implements the duality 
$\mW(\zeta_\free): \mW(\Topf^k(M)) \to \mW(\Topf^{m-k}(M))$ of Remark \ref{rem:dualitiesQuant} 
at the level of GNS representations: 
\begin{align}
U_\free^k\, \pi_\free(\cdot)\, (U^k_\free)^{-1} = \pi_\free \circ \mW(\zeta_\free). 
\end{align}
As a by-product, we obtain that the operators in \eqref{eqn:phase_shift} 
and \eqref{eqn:fluxobs} are intertwined by these isomorphisms. 
In particular, for $m=2k$, this allows us to interpret $U_\free^k$ as the unitary operator on $\scrH^k_\free$ 
that interchanges magnetic and electric fluxes (with a sign that accounts for the appropriate degrees): 
\begin{align}
P^k(\tilde r, (-1)^{k^2+1} r)\, U^k_\free = U^k_\free\, P^k(r, \tilde r). 
\end{align}
Additional natural operations related to duality and in particular to the unitary operator $U^k_\free$ can be defined straightforwardly, cf.\ \cite[Sect.\ 4.2]{Thesis} for a detailed analysis.
\end{remark}

\subsection{\label{sec:ToptState}States for the torsion topological sector}
On the torsion topological sector we introduce a state similar to the one of Remark \ref{rem:faithstatefree}. 
Examples of spacetimes for which this sector is non-trivial are illustrated in \cite{BPhys, FMSa, FMSb}. 

\begin{proposition}\label{prp:statetor}
Let $\omega_\tor:\mW(\Topt^k(M))\to\bC$ be the linear functional specified by
\begin{equation}
\omega_\tor(W(t,\tilde t))=
\begin{cases}
1 & \textrm{if } t=0, \tilde{t}=0,\\
0 & \textrm{otherwise}.
\end{cases}
\end{equation}
Then $\omega_\tor$ is a faithful state on the C$^\ast$-algebra $\mW(\Topt^k(M))$. 
Furthermore, for the state on $\mW(\Topt^k(M))$ and its analogue on $\mW(\Topt^{m-k}(M))$, one has 
\begin{equation}\label{eqn:statetorduality}
\omega_\tor \circ \mW(\zeta_\tor) = \omega_\tor, 
\end{equation}
where $\mW(\zeta_\tor): \mW(\Topt^k(M)) \to \mW(\Topt^{m-k}(M))$ 
denotes the duality introduced in Remark \ref{rem:dualitiesQuant}. 
\end{proposition}

\begin{proof}
Normalization and continuity with respect to $\|\cdot\|_1$ are immediate. 
For positivity, consider $a = \sum_{i \in I} \alpha_i\, W(t_i, \tilde t_i)$, 
where $I$ is a finite set that labels $(t_i, \tilde t_i) \in \Topt^k(M)$ faithfully, 
meaning that $i \neq j$ implies $(t_i, \tilde t_i) \neq (t_j, \tilde t_j)$. Then one finds 
\begin{equation}
\omega_\tor(a^\ast\, a) = \sum_{i \in I} |\alpha_i|^2 \geq 0. 
\end{equation}
Then by the same argument presented in the proof of Proposition \ref{prp:statefree}, 
we obtain the state $\omega_\tor$ on the C$^\ast$-algebra $\mW(\Topt^k(M))$. 
The identity displayed above also shows that the state $\omega_\tor$ is faithful. 
Furthermore, $\omega_\tor$ on $\mW(\Topt^k(M))$ and its analogue on $\mW(\Topt^{m-k}(M))$ 
are clearly related by $\zeta_\tor$. In fact, $(t, \tilde t) = 0 \in \Topt^k(M)$ 
if and only if $\zeta_\tor(t, \tilde t) = 0 \in \Topt^{m-k}(M)$. 
\end{proof}

\begin{remark}
Notice that, passing to the GNS representations associated to $\omega_\tor$ 
on $\mW(\Topt^k(M))$ and on $\mW(\Topt^{m-k}(M))$, we would obtain 
an isomorphism implementing the duality $\zeta_\tor$ between the GNS Hilbert spaces. 
The procedure is identical to the one of Remark \ref{rem:statefree}. 
In fact, whenever the states satisfy a relation such as the one in \eqref{eqn:statetorduality}, 
the above mentioned Hilbert space isomorphism is just a by-product of the GNS construction. 
In particular, for $m=2k$ this leads to duality being unitarily implemented at the GNS level. 
\end{remark}

\noindent We are now in a position to state and prove the main result of Section \ref{Sec:quantstates}: 

\begin{theorem}\label{thmStates}
Let $M$ be an $m$-dimensional globally hyperbolic spacetime $M$ with compact Cauchy surface. 
Then via the factorization of Corollary \ref{isomorphism_dynamical_algebras} we obtain a state 
\begin{align}
\omega: \mW(\Conf^k(M;\bZ)) \longmapsto \bC 
\end{align}
by tensoring the states 
\begin{align}
\omega_\Dyn: \mW(\Dyn^k(M)) \to \bC, && \omega_\free: \mW(\Topf^k(M)) \to \bC, 
&& \omega_\tor: \mW(\Topt^k(M)) \to \bC 
\end{align}
of Propositions \ref{prp:stateDyn}, \ref{prp:statefree} and \ref{prp:statetor}. 
In addition, this construction is compatible with the duality 
$\mW(\zeta): \mW(\Conf^k(M;\bZ)) \to \mW(\Conf^{m-k}(M;\bZ))$ of Remark \ref{rem:dualitiesQuant}, 
namely the states on the source and on the target are related by 
\begin{equation}\label{eqDualityCompatiblity}
\omega \circ \mW(\zeta) = \omega. 
\end{equation}
Furthermore, the state on the dynamical sector $\mW(\Dyn^k(M))$ 
fulfils the microlocal spectrum condition.\footnote{Notice that 
the microlocal spectrum condition only makes sense on $\mW(\Dyn^k(M))$. 
The other sectors only possess ``finitely many'' degrees of freedom. 
In fact, they correspond to group characters on a \textit{finite dimensional} Abelian Lie group 
(the topological configuration space), which is isomorphic to the Cartesian product 
of a finite dimensional torus and a discrete group.} 
\end{theorem}

\begin{proof}
By Corollary \ref{isomorphism_dynamical_algebras} we obtain a state on $\mW(\Conf^k(M;\bZ))$ 
by assigning one on the $\check \otimes$-tensor product 
of the C$^\ast$-algebras $\mW(\Dyn^k(M))$, $\mW(\Topf^k(M))$ and $\mW(\Topt^k(M))$, 
cf.\ \cite{Guichardet} and Section \ref{sec:quantization}. 
By \cite{Guichardet} two commuting representations (one for each factor) on a common Hilbert space 
provide a unique representation of the $\check\otimes$-tensor product. 
Since it is always possible to merge via the tensor product the carrier Hilbert spaces associated to two representations into a single counterpart on which the original representations act on one component and trivially on the other (hence they commute), 
it is sufficient for us to provide a representation of each $\check\otimes$-tensor factor. 
Indeed, this amounts to assigning a state on each sector, 
namely on $\mW(\Dyn^k(M))$, $\mW(\Topf^k(M))$ and $\mW(\Topt^k(M))$ respectively. 
This task is accomplished by Propositions \ref{prp:stateDyn}, \ref{prp:statefree} and \ref{prp:statetor}. 
In particular, Proposition \ref{prp:stateDyn} provides a Hadamard state. 
\vskip .5em

Concerning the behaviour with respect to the duality 
$\mW(\zeta): \mW(\Conf^k(M;\bZ)) \to \mW(\Conf^{m-k}(M;\bZ))$, 
we observe that Propositions \ref{prp:stateDyn}, \ref{prp:statefree} and \ref{prp:statetor} 
provide relations similar to \eqref{eqDualityCompatiblity} for each $\check\otimes$-tensor factor. 
Furthermore, in Remark \ref{rem:dualitiesQuant} we observed that 
the factorization of Corollary \ref{isomorphism_dynamical_algebras} 
intertwines the duality $\mW(\zeta)$ with the $\check\otimes$-tensor product 
of the dualities $\mW(\zeta_\Dyn)$, $\mW(\zeta_\free)$ and $\mW(\zeta_\tor)$. 
Therefore, the claim follows from the definition of $\omega$. 
\end{proof}

\begin{remark}
Although we do not explicitly pursue this goal here, let us mention that our analysis 
can be straightforwardly adapted to the case of self-dual configurations. 
In particular, one obtains an analogue of Theorem \ref{thmStates}. 
However, this requires some care in the presence of torsion. 
In fact, one should keep in mind that the symplectic structure in the self-dual subtheory is not only 
the restriction of the symplectic structure $\sigma$ on $\Conf^k(M;\bZ)$ defined in \eqref{eqn:symplform}, 
but it has to be \textit{rescaled} by $1/2$. This has to be done in order to avoid artificial degeneracies 
in the torsion topological sector that would otherwise show up whenever a $\bZ_2$-factor is present. 
Refer to \cite[Sect.\ 7]{BPhys} for further information about self-dual Abelian gauge fields. 
\end{remark}

\subsection{\label{sec:LorCyl}An example: the Lorentz cylinder}
In the last section, we discuss explicitly a simple but instructive example. Additional ones are present in \cite{Thesis}. We consider the so-called Lorentz cylinder $M = \bR \times \mathbb{S}^1$ (notice that our convention is to set the length, and not the radius, of the circle to $1$). Introducing the standard coordinates $(t, \theta)$, we endow $M$ with the ultra-static metric $g = -\dd t \otimes \dd t+ d\theta \otimes d\theta$. In addition, we focus our attention on the degree $k=1$. Since $H^0(\mathbb{S}^1;\bZ) \simeq \bZ \simeq H^1(\mathbb{S}^1;\bZ)$, it ensues that $\Topt^1(M)$ is trivial, while $\Topf^1(M) \simeq \bT^2 \oplus \bZ^2$. Furthermore, $\Dyn^1(M)=\dd C^\infty(M)\cap \ast \dd C^\infty(M)$. Hence, as a consequence of Corollary \ref{isomorphism_dynamical_algebras}, the C$^\ast$-algebra of observables consists of two factors only:
\begin{equation}
\mW(\Conf^1(M;\bZ)) \simeq \mW(\Dyn^1(M)) \check{\otimes} \mW(\Topf^1(M)).
\end{equation}
A state thereon is completely specified by assigning it independently on each factor of the tensor product. While the one associated to $\mathcal{W}(\Topf^1(M))$ is nothing but \eqref{eqStateTopf}, we can find a more explicit formula for the two-point function on the dynamical sector $\Dyn^1(M)$.
\vskip 1em

Let us take $\rho,\rho^\prime\in\Omega^1_c(M)$. Starting from \eqref{eqn:twopoint-dede}, we consider $\mathfrak{W}_0$. \eqref{eqn:de} shows that $\delta\rho=\delta_{\mathbb{S}^1}\rho_{\mathbb{S}^1}+\partial_t\rho_t$, where we regard $t\mapsto\rho_{\mathbb{S}^1}(t,\cdot)$ and $t\mapsto\rho_t(t,\cdot)$ as smoothly $\bR$-parametrized differential forms on $\mathbb{S}^1$. In addition, recalling \eqref{projections}, it holds that 
\begin{equation}
\pi^0_\perp(\delta\rho) = \sum_{n>0} c_n(t)\, \cos(2\pi n\theta) + d_n(t)\, \sin(2\pi n\theta),
\end{equation}
where
\begin{align}
c_n(t) \doteq 2 \int_{0}^{1} \cos(2\pi n\theta^\prime)\, \delta\rho(t,\theta^\prime)\, d\theta^\prime,
&& d_n(t) \doteq 2 \int_{0}^{1} \sin(2\pi n\theta^\prime)\, \delta\rho(t,\theta^\prime)\, d\theta^\prime.
\end{align}
By writing the same expression for $\rho^\prime$, we can now evaluate directly \eqref{state-part1} obtaining: 
\begin{align}\label{2-pt-Lorentz}
\omega_2([\rho] \otimes [\rho^\prime]) & = \mathfrak{W}_0(\de\rho \otimes \de\rho^\prime) 
	= \mathfrak{W}_\perp(\pi^0_\perp(\delta\rho) \otimes \pi^0_\perp(\delta\rho^\prime)) \nonumber \\
& = \sum_{n>0}\frac{1}{4\pi n} \Big( \widehat{\delta\rho}(n,n)\, \widehat{\delta\rho^\prime}(-n,-n) 
	+ \widehat{\delta\rho}(n,-n)\, \widehat{\delta\rho^\prime}(-n,n) \Big),
\end{align}
where $[\rho],[\rho^\prime]\in\Omega^1_\c(M)_\Dyn$ and 
\begin{equation}
\widehat{\delta\rho}(n,m)=\int_{\bR}dt\,\int_0^1d\theta\,e^{-2\pi i nt} e^{2\pi i m \theta}\delta\rho(t,\theta),
\end{equation}
for all integers $m,n$. 

\begin{remark}
Observe that \eqref{2-pt-Lorentz} and the ensuing $\omega_\Dyn$ identify a ground state for the underlying dynamical theory. At first glance, this might appear as a contradiction to the renowned no-go result for the existence of ground states for a massless scalar field on a two-dimensional globally hyperbolic spacetime, see for example \cite{SHU13}. The origin of such obstruction lies in the presence of an infrared singularity, which is reflected in the contribution of the $0$-mode in the Fourier expansion of the two-point function. It is noteworthy that our implementation of Abelian duality automatically removes such pernicious feature as one can infer by direct inspection of \eqref{2-pt-Lorentz}, where the mode $n=0$ is not present.
\end{remark}


\section*{Acknowledgements}
The authors are grateful to Nicol\`o Drago and Alexander Schenkel 
for stimulating discussions and valuable suggestions. 
M.C.\ and C.D.\ are grateful to the Institute of Mathematics of the University of Potsdam 
for the kind hospitality during the realization of part of this work. 
The work of M.B.\ has been supported by a research fellowship of the Alexander von Humboldt foundation. 
The work of M.C.\ has been partially supported by IUSS (Pavia). 
The work of C.D.\ has been supported by the University of Pavia.


\end{document}